\documentclass[a4paper, 11pt]{article}

\usepackage{a4wide}
\usepackage[utf8]{inputenc}
\usepackage[english]{babel} 
\usepackage{amsmath}
\usepackage{amsfonts}
\usepackage{amssymb}
\usepackage{amsthm}
\usepackage{amscd}
\usepackage{euscript}
\usepackage{tikz}
\usepackage{url}
\usetikzlibrary{topaths,calc}
\usepackage{algorithm,algorithmic}
\usepackage{subfig}

\newtheorem{theo}{Theorem}[section]
\newtheorem{defi}[theo]{Definition}
\newtheorem{prop}[theo]{Proposition}

\newtheorem{lemma}[theo]{Lemma}

\DeclareMathOperator{\rank}{rank}
\DeclareMathOperator{\Ker}{Ker}
\DeclareMathOperator{\im}{Im}

\DeclareMathOperator{\Mat}{Mat}

\newcommand{\C}{\mathbb{C}}

\newcommand{\Z}{\mathbb{Z}}

\newcommand{\F}{\mathbb{F}}

\newcommand{\HH}{\mathbf H}

\newcommand\mc[1]{\mathcal{#1}}

\newcommand{\dha}{\partial_{1}^{\mathcal H}}
\newcommand{\dhb}{\partial_{2}^{\mathcal H}}

\title{Decoding color codes by projection onto surface codes}

\author{
    Nicolas Delfosse\\
    delfosse@lix.polytechnique.fr\\
    LIX \& INRIA Saclay - \'Ecole Polytechnique, 91128 Palaiseau, France\\
}

\begin{document}

\maketitle

\begin{abstract}
We propose a new strategy to decode color codes, which is based on the projection of the error onto three surface codes. This provides a method to transform every decoding algorithm of surface codes into a decoding algorithm of color codes. Applying this idea to a family of hexagonal color codes, with the perfect matching decoding algorithm for the three corresponding surface codes, we find a phase error threshold of approximately $8.7\%$.
Finally, our approach enables us to establish a general lower bound on the error threshold of a family of color codes depending on the threshold of the three corresponding surface codes.
These results are based on a chain complex interpretation of surface codes and color codes.
\end{abstract}

\section{Introduction}

Classical Low Density Parity-Check (LDPC) codes are of high practical interest for classical error correction because they are both high rate and endowed with an efficient decoding algorithm. It is therefore natural to investigate their quantum generalizations \cite{MMM04}.
Topological constructions of quantum codes lead to some of the most promising families of quantum LDPC codes including Kitaev's toric code \cite{Ki97}, surface codes \cite{BM06, Ze09, De13}, color codes \cite{BM07a} and other geometrical constructions \cite{CDZ12, TZ09, KP13, FH13, Au13}.
Besides their appeal for quantum error-correction, topological codes also exhibit some interesting features for fault tolerant quantum computing, such as the property that some gates can be implemented topologically. In this paper, we are interested in decoding algorithms for topological codes, which is one essential ingredient of fault tolerance to avoid error accumulation during a computation.

The most natural way to decode quantum LDPC codes is to adapt classical strategies \cite{PC08}. Unfortunately, the unavoidable presence of short cycles in the Tanner graph of CSS codes makes this iterative decoding algorithm difficult to adapt to the quantum setting.
%Some more sophisticated classical techniques have been imported in quantum information recently \cite{AMT12a, AMT12b}.
We focus here on the decoding problem for two families of quantum LDPC codes based on tilings of surfaces: surface codes and color codes. In this special case, different strategies making use of the geometrical structure of the code have been proposed.

Topological codes appeared with Kitaev's toric code defined from a square tiling of the torus \cite{Ki97}. More generally, we can associate a quantum code to every tiling of surface \cite{BM07a}. These surface codes form a family of quantum LDPC codes because they are defined by stabilizers acting on a small number of qubits but they are also physically local, reducing the error during the syndrome measurement.
A first approach based on a perfect matching algorithm was proposed by Dennis, Kitaev, Landahl and Preskill to decode surface codes \cite{DKLP02}. Their algorithm computes a most likely error using a minimum weight perfect matching algorithm. It has been applied to fault-tolerant quantum computing by Wang, Fowler, Stephen and Holenberg \cite{WFSH10}.
More recently, inspired by ideas from statistical physics, Duclos-Cianci and Poulin introduced a renormalization group decoding algorithm, for the toric code \cite{DP10}.

Color codes, introduced by Bombin and Martin-Delgado, are also derived from tilings of surfaces \cite{BM06}. For fault-tolerance, they offer an advantage over surface codes in that we can apply, topologically, a larger number of operations on the encoded qubits. The perfect matching decoding algorithm has been generalized to color codes by Wang, Fowler, Hill and Hollenberg, but, the decoding problem becomes more difficult with color codes because it corresponds to a hypergraph matching \cite{WFHH10}. A message passing algorithm which shares some characteristics of the renormalization group decoding algorithm has been proposed by Sarvepalli and Raussendorf \cite{SR12}. The fault tolerance of a family of color codes based on the square-octogonal lattices is studied by Landahl, Andersen and Rice in \cite{LAR11}.

In the present work, we relate the decoding problem for color codes to the decoding problem for surface codes. This is done by the projection of an error acting on a color code onto three surface codes. This strategy allows us to transform every decoding algorithm for surface codes into a decoding algorithm for color codes.
As an example, we consider the decoding algorithm for color codes deduced from the perfect matching decoding algorithm for surface codes. Our numerical results on the hexagonal color codes exhibit a phase error threshold of $8.7\%$, higher than the threshold of $7.8\%$ observed by Sarvepalli and Raussendorf for the same codes \cite{SR12}.

Our main tool is a chain complex interpretation of CSS codes and topological codes. A similar approach has also been used recently in quantum information by Freedman and Hastings \cite{FH13} and by Audoux \cite{Au13}.
The projection from the color code to the surface codes is a morphism of chain complexes. Therefore, it conserves the whole structure of the chain complex representing the color code.

The optimal error threshold of some particular families of surface codes and color codes was previously estimated using a mapping onto a Ising model \cite{DKLP02, KBM09, Oh09, BAOKM12}.
Our approach allows us to compare the thresholds of color codes and surface codes achievable using an efficient decoding algorithm. Using the fact that the projection onto the surface codes conserves the error model, \emph{i.e.} the depolarizing channel for a color code is sent onto a depolarizing channel over the surface codes, we derive a general lower bound on the threshold of a family of color codes as a function of the threshold of the three corresponding surface codes.

The remainder of this article is organized as follows. The definition of CSS codes and the error model are recalled in Section~\ref{section:background}. Section~\ref{section:complexes} is devoted to the construction of CSS codes and surface codes from chain complexes. This point of view will be essential to our study of color codes.
%We recall the construction of surface codes using the formalism of chain complexes in Section~\ref{section:surface}.
We describe color codes with the language of chain complexes in Section~\ref{section:color}. Then we introduce a projection from the color code chain complex to the surface code chain complex in Section~\ref{section:projection}. Our decoding algorithm for color codes, and the analysis of its performance are presented in Section~\ref{section:decoding}. Figure~\ref{fig:couleur_correction} represents the different steps of our decoding algorithm in a hexagonal color code. We will refer to this example throughout this article.

\newpage
\section{Background} \label{section:background}

\subsection{The CSS construction of quantum codes} \label{section:CSS}

In this section, we review the combinatorial construction of CSS codes, their syndrome function and the noise model.

A stabilizer code of length $n$ is defined to be the set of fixed point of a family of commuting Pauli operators $S_1, S_2, \dots, S_r$ acting over $n$ qubits (that is over $(\C^2)^{\otimes n})$.
When this family is composed of $r_X$ operators of $\{I,X\}^{\otimes n}$ and $r_Z$ operators of $\{I, Z\}^{\otimes n}$, it suffices to check the commutation relations between operators of different kind. Recall that $I$ is the identity matrix of size 2, $X$ is the Pauli matrix representing the bit flip error and $Z$ is the Pauli matrix representing the phase error. This construction has been introduced by Calderbank, Shor \cite{CS96} and Steane \cite{St96}. It is the family of \emph{CSS codes}. The isomorphisms between $\{I, X\}$, $\{I, Z\}$ and the field $\F_2$ enables us to relate classical and quantum codes. By these isomorphisms, the stabilizers $S_i$ correspond to binary vectors and the commutation relation corresponds to the orthogonality relation in $\F_2^n$. This leads to a combinatorial description of CSS codes that is recalled in the next paragraph.

%\vspace{.2cm}
%\noindent
%\textbf{CSS codes:}
A CSS code is defined by two binary matrices $\HH_X \in \mc M_{r_X, n}(\F_2)$ and $\HH_Z \in \mc M_{r_Z, n}(\F_2)$ satisfying $\HH_X {\HH_Z}^t=0$, \emph{i.e.} with orthogonality between rows of $\HH_X$ and rows of $\HH_Z$.
The rows of these two matrices correspond to the stabilizers $S_i$ in the previous definition. The integer $n$ is the \emph{length} of the quantum code, it is the number of qubits used to describe the encoded states. The number of encoded qubits is $k=n-\rank \HH_X- \rank \HH_Z$.
Two classical codes are associated with such a quantum code.
Denote by $C_X$ the kernel of the matrix $\HH_X$ and denote by $C_Z$ the kernel of $\HH_Z$. Remark that $C_X^\perp$ is the space generated by the rows of $\HH_X$ and $C_Z^\perp$ is the space generated by the rows of $\HH_Z$.
%From the orthogonality relations, the spaces $C_X$ and $C_Z$ satisfy $C_X^\perp \subset C_Z$ and $C_Z^\perp \subset C_X$.

%\vspace{.2cm}
%\noindent
%\textbf{Quantum errors:}
A \emph{quantum error} acting on this CSS code is defined to be a pair of vectors $E=(E_X, E_Z) \in \F_2^n \times \F_2^n$. It simply corresponds to a Pauli operator, acting over $n$ qubits, decomposed following $X$ and $Z$.
%\vspace{.2cm}
%\noindent
%\textbf{Syndrome of a quantum error:}
When an error $(E_X, E_Z)$ corrupts a quantum state of the code, we can measure its \emph{syndrome}. It is the pair of vectors $(s_X, s_Z) \in \F_2^{r_Z} \times \F_2^{r_X}$ defined by $s_X = \HH_Z E_X^t$ and $s_Z = \HH_X E_Z^t$.
The syndrome can be measured and it is the only information that we possess about the error which occurs.

%\vspace{.2cm}
%\noindent
%\textbf{Degeneracy:}
By construction of stabilizer codes and CSS codes, a group of errors fixes the quantum code. It is the set of errors sastifying $E_X \in C_X^\perp \text{ and } E_Z \in C_Z^\perp$. This set of operators corresponds to the subgroup, called the stabilizer group, generated by the operators $S_i$.
As a consequence, all the errors of the coset $(E_X+C_X^\perp , E_Z + C_Z^\perp)$ behave like $(E_X, E_Z)$. This property of stabilizer codes, referred to as the \emph{degeneracy}, allows us to correct all the errors of this coset using the same procedure.

The properties of errors for CSS codes are summarized in Table~\ref{tab:CSS_codes}.
\begin{table}[hb]
\centering
\footnotesize
\caption{Error representation for CSS codes.}
\label{tab:CSS_codes}
%\centerline{\footnotesize\smalllineskip
\begin{tabular}{c c c}
\hline
 & $X$-component & $Z$-component\\
\hline
\text{error} & $E_X \in \F_2^n$ & $E_Z \in \F_2^n$\\
syndrome & $s_X = \HH_Z E_X^t \in \F_2^{r_Z}$ & $s_Z = \HH_X E_Z^t \in \F_2^{r_X}$ \\
stabilizer & $E_X \in C_X^\perp$ & $E_Z \in C_Z^\perp$ \\
\hline
\end{tabular}
\end{table}

%\vspace{.2cm}
%\noindent
%\textbf{Error model:}
Let us now introduce the error model. We consider independently the two components $E_X$ and $E_Z$ of a quantum error and we assume that these two vectors of $\F_2^n$ are subjected to a binary symmetric channel of probability $p$. It is a simplified version of the depolarizing channel of probability $3p/2$. We assume that we are able to measure the syndrome without measurement errors.

Our goal is to compute a most likely error coset $(\tilde E_X+C_X^\perp, \tilde E_Z+C_Z^\perp)$ reaching a given syndrome $(s_X, s_Z)$. To simplify, we often look for a most likely error $(\tilde E_X, \tilde E_Z)$ reaching a given syndrome.
The  $X$-component and the $Z$-component of the error can be considered separately.

%In the case of surface codes and color codes, the two components of the error have the same behavior so only need to consider the phase component $E_Z$. \emph{Our goal will be to compute a most likely error $\tilde x$ reaching a given syndrome $s_Z$}, that is a vector $\tilde x \in \F_2^n$ of minimum weight such that $\HH_X \tilde x^t=s_Z$ for a given syndrome $s_Z$.

\subsection{Tilings of surfaces} \label{section:tilings}

The quantum codes studied in this paper are based on topological objects, that we introduce in this part.

A \emph{graph} is defined as a pair $(V, E)$, where $V$ is a set and $E$ is a set of pairs of elements of $V$. The elements of $V$ are called \emph{vertices} and the elements of $E$  are called \emph{edges}. The two vertices included in an edge $e=\{u, v\}$ are the \emph{endpoints} of $e$. The edge $e$ is said to be incident to the two vertices $u$ and $v$.
A \emph{hypergraph} is defined similarly but we allow edges to contain more than two vertices.

A \emph{tiling of surface} is a triple $G=(V, E, F)$, where $(V, E)$ is a finite graph embedded in a compact 2-manifold (surface), without boundary, without overlapping edges, and $F$ is the set of faces defined by this embedding. A face is given as the set of edges on its boundary. For example a face $f = \{e_1, e_2, \dots, e_m\}$ contains $m$ edges. This face $f$ is said to be incident to the edge $e_i$, for all $i$ such that $1\leq i \leq m$.
Without loss of generality, we can assume that this surface is smooth and connected. We also assume that the graph $(V, E)$ contains neither loops nor multiple edges.

Given a tiling of surface $G=(V, E, F)$, we construct its \emph{dual tiling} $G^*=(V^*, E^*, F^*)$. It is the graph of vertex set $V^*=F$, such that two vertices of $V^*$ are joined by an edge if and only if the corresponding faces share an edge in the graph $G$.
In other word, every edge of $G$ corresponds to an edge of its dual $G^*$. The set of edges incident to a vertex $v$ of $G$ induces a face of the dual tiling. This leads to a one-to-one correspondence between the faces of the dual graph $G^*$ and the vertices of the graph $G$.
In some degenerated cases, this dual graph could have loops or multiple edges. We assume that this does not happen here.

\section{Construction of CSS codes from chain complexes} \label{section:complexes}

The orthogonality relations required to define CSS codes can be deduced from the properties of chain complexes \cite{La02}. In this section, we describe these quantum codes with the language of chain complexes. This formalism will be essential to study color codes and to decompose their decoding problem.

\subsection{Definition of 2-complexes}

Let us recall the definition of a 2-complex.

\begin{defi} \label{defi:2-complex}
A \emph{2-complex} is a sequence of three $\F_2$-vector spaces $C_0, C_1, C_2$, endowed with two $\F_2$-linear applications $\partial_2 : C_2 \rightarrow C_1$ and $\partial_1: C_1\rightarrow C_0$ such that:
\begin{equation}
\label{eqn:complex_composition}
\partial_1 \circ \partial_2 = 0.
\end{equation}
\end{defi}

The applications $\partial_1$ and $\partial_2$ are called the \emph{boundary applications}. The vectors of the space $\Ker \partial_i$ are called $i$-cycles or \emph{cycles} and the vectors of the space $\im \partial_i$ are called $i$-boundaries or \emph{boundaries}. In what follows, we consider only finite dimensional spaces $C_i$.
%From equation (\ref{eqn:complex_composition}), we deduce the inclusion $\im \partial_2 \subset \Ker \partial_1$. It is then possible to consider the quotient space $\Ker \partial_1 / \im \partial_1$. It is the \emph{first homology group} of the 2-complex and it is denoted by $H^1$.

\vspace{.2cm}
To define a chain complex from a tiling of surface, we will use the $\F_2$-linear structure of the power set $\mathcal P(X)$ of a finite set $X$. The power set $\mathcal P(X)$ is naturally a $\F_2$-linear vector space for the symmetric difference. Alternatively, it is can be represented as the $\F_2$-linear space of formal sums of elements of $X$:
$$
\bigoplus_{x \in X} \F_2 x = \{\sum_{x \in X} \lambda_x x \ | \ \lambda_x \in \F_2 \},
$$
endowed with the componentwise addition:
$\sum_x \lambda_x x + \sum_x \mu_x x = \sum_x (\lambda_x + \mu_x) x$.
A subset $Y \subset X$ corresponds to the vector $\sum_{y \in Y} y$ in the space $\oplus_{x \in X} \F_2 x$. Inversely, a vector $\sum_x \lambda_x x$ of the formal sum is the indicator vector of the subest $Y = \{ x \in X \ | \ \lambda_x=1 \}$.
Remark that the weight of $Y$, regarded as a binary vector, corresponds to the cardinality of the set $Y$. In what follows, it will be convenient to regard a subset as a binary vector.
The space $\mc P(X)$ is a $\F_2$-linear space of dimension $|X|$. A canonical basis is given by the set of all the singleton $\{x\} = Y$ included in $X$. We abusively say that the elements of $X$ form a basis of this space.

\vspace{.2cm}
In order to study and classify manifolds, different chain complexes have been introduced. For example, the \emph{cellular homology complex} associated with a tiling of surface $G=(V, E, F)$ is the chain complex defined on the spaces
$$
C_2 = \bigoplus_{f \in F} \F_2 f,
\qquad C_1 = \bigoplus_{e \in E} \F_2 e,
\qquad C_0 = \bigoplus_{v \in V} \F_2 v
$$
and the boundary maps $\partial_2$ and $\partial_1$, which are the $\F_2$-linear applications
$$
C_2 \overset{\partial_2}{\longrightarrow} C_1 \overset{\partial_1}{\longrightarrow} C_0
$$
such that $\partial_2(f) = \sum_{e \in f} e$ and $\partial_1(e) = \sum_{v \in e} v$.
The canonical basis of the space $C_2$ (respectively $C_1$ and $C_0$) is composed of the faces $f \in F$ (respectively the edges and the vertices) of the tiling $G$. That is to say, $C_2$ (respectively $C_1$ and $C_0$) is the power set of $F$ (respectively $E$ and $V$).
Geometrically, the application $\partial_2$ sends a face onto the set of edges on its topological boundary and the application $\partial_1$ sends an edge onto its endpoints. This explains the term boundary applications. Remark that the boundary of a path (\emph{i.e.} its image under $\partial_1$) is composed of its two \emph{terminal points}. We can easily check that the composition of the two boundary applications is zero. In the next section, we introduce the surface codes from this 2-complex.

\subsection{The CSS code associated with a 2-complex} \label{section:complex_codes}

From the structure of chain complexes, we immediately obtain the orthogonality relations needed to define a CSS code:
\begin{prop} \label{prop:complex_code}
Every 2-complex $C_2 \overset{\partial_2}{\rightarrow} C_1 \overset{\partial_1}{\rightarrow} C_0$, based on finite dimensional spaces, defines a CSS code of length $n = \dim C_1$, based on the matrices $\HH_X = \Mat(\partial_1, B_1, B_0)$ and $\HH_Z = \Mat(\partial_2, B_2, B_1)^t$, where $B_i$ is a basis of the space $C_i$.
\end{prop}

\begin{proof}
Equation (\ref{eqn:complex_composition}) can be translated matricially into $\HH_X {\HH_Z}^t = 0$ which proves the orthogonality relations. The length of the resulting quantum code is given by the dimension of $C_1$ because it is the number of columns of the matrices $\HH_X$ and $\HH_Z$.
\end{proof}

The properties of such a CSS code depend on the 2-complex but also on the \emph{dual complex}. Let us recall the definition of this chain complex.
It is the 2-complex defined on the dual spaces $C_i^*$ (the space of $\F_2$-linear forms over $C_i$), and endowed with the transposed applications $\partial_i^* : C_{i-1}^* \rightarrow C_i^*$, which sends $\phi \in C_{i-1}^*$ onto $\partial_i^*(\phi) = \phi \circ \partial_i$. This leads to the dual complex
$$
C_2^* \overset{\partial_2^*}{\longleftarrow} C_1^* \overset{\partial_1^*}{\longleftarrow} C_0^*.
$$
In the dual basis, the matrix of the application $\partial_1^*$ is the matrix $\HH_X^t$ and the matrix of the application $\partial_2^*$ is the matrix $\HH_Z$.  Based on this representation of the matrices defining a CSS code, Table~\ref{tab:complex_codes} describe the errors, the syndrome function and the stabilizers with the language of chain complexes. It is the translation of Table~\ref{tab:CSS_codes}.
\begin{table}[h]
\centering
\footnotesize
\caption{Error representation for 2-complex codes.}
\label{tab:complex_codes}
%\centerline{\footnotesize\smalllineskip
\begin{tabular}{c c c}
\hline
 & $X$-component & $Z$-component\\
\hline
\text{error} & $E_X \in C_1^*$ & $E_Z \in C_1$\\
syndrome & $s_X = \partial_2^*(E_X) \in C_2^*$ & $s_Z = \partial_1(E_Z) \in C_0$ \\
stabilizer & $E_X \in \im \partial_1^*$ & $E_Z \in \im \partial_2$ \\
\hline
\end{tabular}
\end{table}

%\begin{fact} \label{fact:complex_codes}
%Given a CSS associated with a 2-complex $C_2 \overset{\partial_2}{\rightarrow} C_1 \overset{\partial_1}{\rightarrow} C_0$, we have:
%\begin{itemize}
%\item An error is a pair $(E_X, E_Z) \in C_1^* \times C_1$.
%\item Its syndrome is $(s_X, s_Z) \in C_2^* \times C_0$ with $s_X = \ \partial_2^*(E_X)$ and $s_Z = \partial_1(E_Z)$.
%\item A stabilizer is an error with $E_X \in \im \partial_1^* \text{ and } E_Z \in \im \partial_2$.
%\end{itemize}
%\end{fact}

Using the linear isomorphism $C_i \simeq C_i^*$, which conserves the weight, this dual complex can be regarded in the spaces $C_i$:
$$
C_2 \overset{\partial_2^*}{\longleftarrow} C_1 \overset{\partial_1^*}{\longleftarrow} C_0.
$$
More precisely, remark that the space $C_i$, of basis $B_i$, can be seen as the power set of $B_i$, using the structure introduced in the previous section. This allows us to define the $\F_2$-linear application $\partial_i^*$ by
\begin{align}
\label{eqn:transposed_application}
\partial_i^*(x) = \sum_{\substack{y \in B_i \\ x \in \partial_i(y)}} y,
\end{align}
for all $x \in B_{i-1}$.

\vspace{.2cm}
Let us apply this construction to a cellular homology complex. From Proposition~\ref{prop:complex_code}, the cellular homology complex associated with a tiling $G=(V, E, F)$ defines a CSS code of length $n=|E|$. It is the \emph{surface code} associated with $G$. To see that this construction of surface codes coincides with the original definition of Kitaev \cite{Ki97}, and Bombin and Martin-Delgado \cite{BM07a}, it suffices to remark that the stabilizers, given in Table~\ref{tab:complex_codes}, correspond to the plaquette operators and the site operators introduced by Kitaev.
By definition of the cellular homology complex, the space $\im \partial_2$ is generated by the faces of the graph. These stabilizers correspond to the plaquette operators. We can see that the space $\im \partial_1^*$ corresponds to the site operators by using the definition of the application $\partial_1^*$ given in Equation~(\ref{eqn:transposed_application}). Thus, the corresponding stabilizer group is generated by the face operators in $Z$ and the site operators in $X$.

\subsection{The decoding problem for surface codes} \label{section:surface}

Before coming to a detailed study of the decoding problem for color codes, let us recall the basic ideas involved in the decoding algorithm for surface codes.

By symmetry between the cellular homology complex and its dual, it is sufficient to consider the $Z$-component of the error. Indeed, in the special case of cellular homology complexes, the dual complex is the cellular homology complex associated with the dual graph~\cite{Ha02}:
$$
C_2(G^*) \overset{\partial_2^{G^*}}{\rightarrow} C_1(G^*) \overset{\partial_1^{G^*}}{\rightarrow} C_0(G^*).
$$
%Thus the $X$-component of the error can be treated similarly in the dual graph $G^*$.
This leads to the simplified framework presented in Table~\ref{tab:surface_codes}. In this table, we consider only the $Z$-component of the error, denoted by $x$. We give a chain complex point of view and a graphical point of view. To work with the $X$-component of the error, replace the graph $G$ by its dual $G^*$.
%\begin{fact} \label{fact:surface_codes}
%Given a surface codes associated to a tiling $G=(V, E, F)$, we have:
%\begin{itemize}
%\item An error is a vector $x \in C_1$, equivalently, it is a subset of the edge set $x \subset E$.
%\item Its syndrome is the vector $s = \partial_1(x)$, equivalently, it is the set of terminal points of $x$.
%\item A stabilizer is an error $x \in \im \partial_2$, equivalently, it is a sum of faces.
%\end{itemize}
%\end{fact}
\begin{table}[h]
\centering
\footnotesize
\caption{Error representation for surface codes.}
\label{tab:surface_codes}
%\centerline{\footnotesize\smalllineskip
\begin{tabular}{c c c}
\hline
 & 2-complex point of view & graphical point of view\\
\hline
\text{error} & $x \in C_1$ & $x \subset E$ is a subset of the edge set\\
syndrome & $s = \partial_1(x) \in C_0$ & $s \subset V$ is the set of terminal vertices of $x$\\
stabilizer & $x \in \im \partial_2$ & $x$ is a boundary \\
\hline
\end{tabular}
\end{table}

%First, we focus on the component $E_Z \in C_1$ of a quantum error acting on this surface code. By definition of $C_1$, this vector can be seen as a set of edges of the graph $G$. Its syndrome is the vector $s_Z=\partial_1(E_Z)$ of $C_0$, which corresponds to the its terminal points in the graph $G$.
Hence, our goal is to recover a set of edges $x$, from the knowledge of its set of terminal vertices $s$, and it is sufficient to determine $x$ up to a boundary.
Recall that the terminal vertices of a set of edges are the vertices which are reached an odd number of times by an edge. To get a most likely error, that is to say an error of minimum weight, it suffices to choose this set $x$ with minimum cardinality.
%Note that what we abusively call an error in Fact~\ref{fact:surface_codes}, is in fact only a component of a Pauli error.

\begin{figure}[htbp]
\centering
\includegraphics[scale=1]{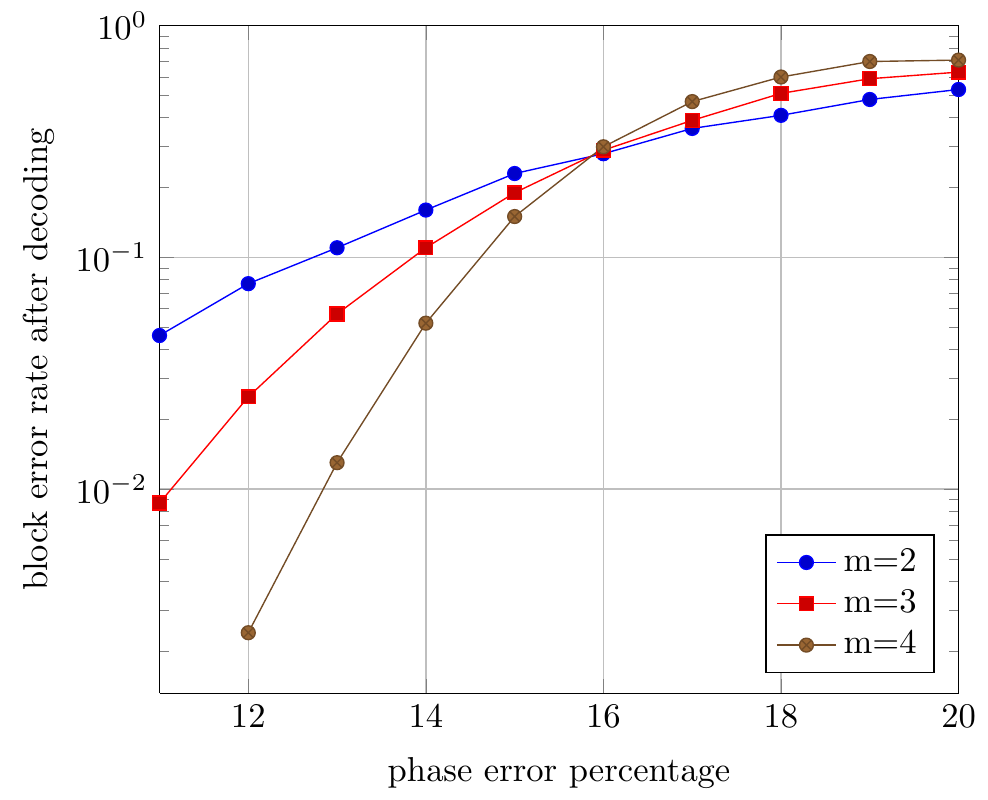}
  \caption{Bloc error rate after decoding as a function of the phase error rate, for hexagonal toric codes of length $9.2^{2m}$.}
  \label{fig:perf_surface}
\end{figure}

The perfect matching decoding algorithm \cite{DKLP02} determines a most likely error by matching pairs of vertices of the syndrome $s$, using Edmonds' minimum weight perfect matching algorithm \cite{Ed65a, Ed65b}. To illustrate the performance of this decoding algorithm, we consider a family of surface codes based on hexagonal tilings of the torus. Our results are represented in Figure~\ref{fig:perf_surface}. We observe a threshold closed to $15.9\%$. This threshold will be related to the performance of a family of color codes in section~\ref{subsection:threshold}. We use the implementation ''blossom~V'' of the minimum weight perfect matching algorithm, due to Kolmogorov \cite{Ko09}. The worst case complexity of this decoding algorithm is in $O(|V|^3 |E|)$, but the typical complexity is definitely better. An improved version of the perfect matching decoding algorithm, with a better complexity, has been proposed recently by Fowler, Whiteside and Hollenberg \cite{FWH12}.

The rest of this paper is devoted to the study of the decoding problem for color codes.

\section{A 2-complex definition of color codes} \label{section:color}

In this section, we describe color codes with the language of chain complexes. %This point of view will allow us establish a connection between surface codes and color codes.
%Our strategy is illustrated in Figure~\ref{fig:couleur_correction} where the different steps of our decoding algorithm are described in a hexagonal color code. We will refer to this example throughout this section.
Color codes are usually defined from a 3-regular tiling of surface $G$, whose faces are 3-colorable \cite{BM06}. Recall that the faces of a graph are said to be 3-colorable if and only if there exists a 3-coloration of the faces such that two faces sharing an edge do not support the same color.
We will consider a chain complex definition of color codes based on the dual tiling $G^*$. This point of view will be more appropriate for our decomposition of the decoding problem.

\begin{figure}[htbp]
\centering
\includegraphics{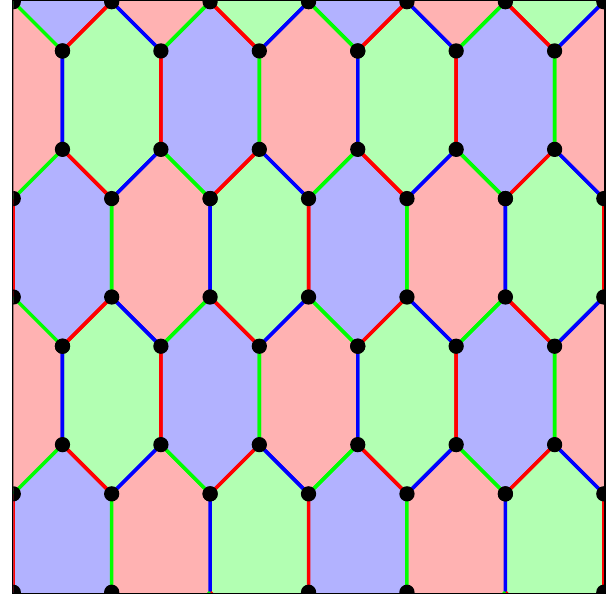}
\includegraphics{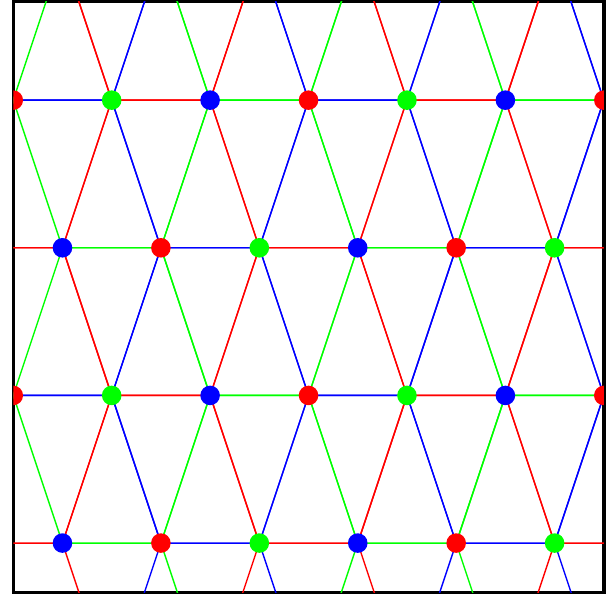}
\caption{A tiling of the torus with 3-colorable faces and its dual. The vertices of the dual inherit of the 3-coloration.}
\label{fig:colored_tiling_dual}
\end{figure}

In what follows, $G=(V, E, F)$ is a 3-regular tiling of surface. We suppose that the faces of this graph are 3-colorable. As a consequence, the dual graph $G^*=(V^*,E^*,F^*)$ is equipped with triangle faces and its vertices are 3-colorable (without neighbors supporting the same color). We assume that, these graphs, $G$ and $G^*$, have neither loop nor multiple edge.
An example of hexagonal tiling with 3-colored faces and its dual are represented in Figure~\ref{fig:colored_tiling_dual}.

\subsection{Definition of the hypergraph}

Our first objective is to recall the definition of color codes with the formalism of chain complexes. This is done by considering the dual tiling $G^*$ as a hypergraph~$\mathcal H$.
By definition, a hypergraph is a pair $(\mc V, \mc E)$ where $\mc V$ is a finite vertex set and $\mc E$ is composed of subsets of $\mc V$ (not necessarily of cardinality 2). Just as in the definition of a tiling of surface, the hypergraph is equipped of a set of faces $\mc F$, which is composed of subsets of $\mc E$.

\begin{defi} \label{defi:hypergraphe}
The {\em hypergraph associated with $G^*$} is defined to be the hypergraph
$\mathcal H=(\mathcal V, \mathcal E, \mc F)$ such that:
\begin{itemize}
\item The vertex set $\mc V$ of $\mc H$ is equal to the vertex set $V^*$ of $G^*$.
\item The hyperedges are the triples of vertices $\mathcal E(f)$ included in a face $f$ of $G^*$. The set of hyperedges is $\mc E = \{\mc E(f) \ | \ f \in F^*\}$.
\item The hyperfaces are the sets $\mc F(v)$ of hyperedges incident to a vertex $v$ of $G^*$. The set of hyperfaces is $\mc F = \{\mc F(v) \ | \ v \in V^* \}$.
%For every vertex $v$ of $G^*$, denote by $\mathcal F(v)$ the set of hyperedges of $\mc H$ incident to $v$. The set of hyperfaces of $\mc H$ is $\mathcal F = \{\mathcal F(v) \ | \ v \in V^*\}$.
\end{itemize}
\end{defi}

Since the graph $G^*$ is chosen to have triangle faces, the hyperedges always contain 3 vertices. That means that $\mc H$ is a 3-uniform hypergraph.
These 3 vertices being neighbors, their colors are different.
More precisely, every face $f$ of the graph $G^*$ is a triple of edges
$f = \{ \{u_R, u_G\}, \{u_G, u_B\}, \{u_B, u_R\} \}$,
such that $u_{\bf c}\in V^* = \mc V$ is a vertex of color ${\bf c}$.
The hyperedge $\mc E(f) \in\mc E$ associated with this face is 
\begin{equation} \label{eqn:hyperedges}
\mc E(f) =\{u_R, u_G, u_B\}.
\end{equation}

When the graph $G^*$ is $m$-regular, all the hyperfaces contain $m$ hyperedges.
Let us describe these hyperfaces.
The tiling $G^*$ is a triangulation, therefore the vertex $v$ has $m$ neighbors $v_1, v_2, \dots, v_m$, ordered such that $v_i$ and $v_{i+1}$, and $v_m$ and $v_1$, are neighbors. The $m$ hyperedges of $\mc H$, containing the vertex $u$ are given by $e_i= \{u, v_i, v_{i+1}\}$, for $1 \leq i \leq m-1$, and $e_m = \{u, v_m, v_1\}$. The hyperface $\mc F(u)$ is
\begin{equation} \label{eqn:hyperfaces}
\mc F(u) = \{e_1, e_2, \dots, e_m \}.
\end{equation}

The hypergraph associated with the hexagonal tiling of Figure~\ref{fig:colored_tiling_dual} is represented in Figure~\ref{fig:colored_tiling_hyperface}. Its hyperfaces are composed of 6 hyperedges.

\subsection{Definition of color codes from hypergraph 2-complexes}

Using the 2-dimensional structure of the hypergraph $\mc H$, we can define a 2-complex:
\begin{defi} \label{defi:hypergraph_complex}
The {\em 2-complex associated with the hypergraph $\mc H$} is defined to be the complex composed of the three spaces
$$
C_2(\mc H) = \bigoplus_{f \in \mc F} \F_2 f,
\qquad C_1(\mc H) = \bigoplus_{e \in \mc E} \F_2 e,
\qquad C_0(\mc H) = \bigoplus_{v \in \mc V} \F_2 v
$$
and the boundary maps $\partial_2^{\mc H}$ and $\partial_1^{\mc H}$, which are the $\F_2$-linear applications
$$
C_2(\mc H) \overset{\partial_2^{\mc H}}{\longrightarrow} C_1(\mc H) \overset{\partial_1^{\mc H}}{\longrightarrow} C_0(\mc H)
$$
such that $\partial_2^{\mc H}(f) = \sum_{e \in f} e$, for all $f \in \mc F$ and $\partial_1^{\mc H}(e) = \sum_{v \in e} v$, for all $e \in \mc E$.
\end{defi}

From Equation~(\ref{eqn:hyperedges}), every hyperedge $e$ of $\mc H$ can be write $e = \{u_R, u_G, u_B\}$, such that $u_{\bf c}$ is a vertex of color $\bf c$. The boundary of this hyperedge is
\begin{equation} \label{eqn:partial_1_H}
\partial_1^{\mc H}(e) = u_R+u_G+u_B \in C_0(\mc H).
\end{equation}
Using the notation of Equation~(\ref{eqn:hyperfaces}), the boundary of a hyperface $f = \{e_1, e_2, \dots, e_m\}$ is the sum
\begin{equation} \label{eqn:partial_2_H}
\partial_2^{\mc H}(f) = \sum_{i=1}^m e_i \in C_1(\mc H).
\end{equation}

As in the case of the cellular homology complex defined in Section~\ref{section:complexes}, a vector of $C_2(\mc H)$ (respectively $C_1(\mc H)$ and $C_0(\mc H)$) can be regarded as a subset of $\mc F$ (respectively $\mc E$ and $\mc V$).
The image of a subset of $\mc F$, by the application $\partial_2^{\mc H}$, is the set of hyperedges included in an odd number of hyperfaces this subset. Similarly, the image of a subset of hyperedge $x \subset \mc E$, by the application $\partial_1^{\mc H}$, is the set of vertices of $\mc H$ that are contained in an odd number of hyperedges of $x$. By analogy with the cellular homology complex, these vertices are called the \emph{terminal vertices} of $x$.

\begin{figure}[htbp]
\centering
\includegraphics{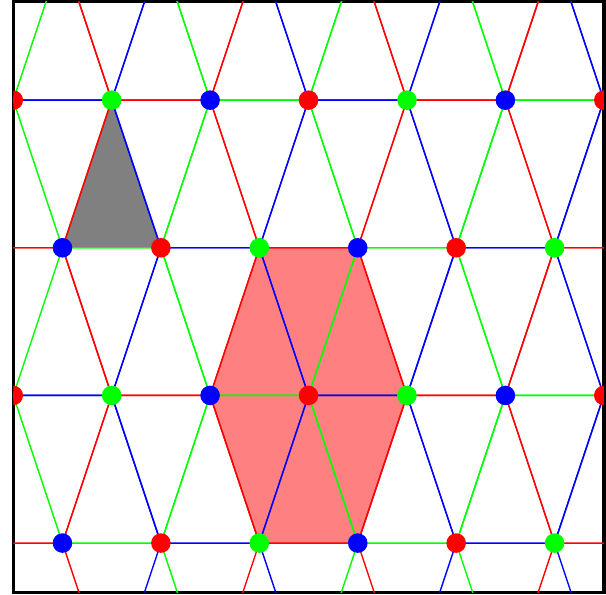}
\caption{In grey, a hyperedge of the hypergraph $\mc H$ associated with an hexagonal tiling of the torus. In red, a hyperface of this hypergraph.}
\label{fig:colored_tiling_hyperface}
\end{figure}

This sequence of vector spaces is a chain complex because the composition of the boundary maps $\dha$ and $\dhb$ is zero.
Indeed, it suffices to check that the image of a hyperface $\mc F(v)$ by $\dha\circ\dhb$ is zero. This proves that $\dha\circ\dhb=0$ by linearity.
From Proposition~\ref{prop:complex_code}, this relation enables us to define a CSS code. This quantum code is called the \emph{color code} associated with the graph $G$. This definition is equivalent to the original definition of color codes \cite{BM06}.

\subsection{The decoding problem for color codes}

In this part, we present the decoding problem for color codes. In what follows, the sequence $$C_2(\mc H) \overset{\partial_2^{\mc H}}{\longrightarrow} C_1(\mc H) \overset{\partial_1^{\mc H}}{\longrightarrow} C_0(\mc H)$$
denotes a 2-complex defining a color code.

As in the case of surface codes, the decoding problem for surface codes can be simplified by exploiting the symmetries between the 2-complex and its dual.
\begin{lemma} \label{lemma:color_duality}
The dual complex of 
$C_2(\mc H) \overset{\partial_2^{\mc H}}{\longrightarrow} C_1(\mc H) \overset{\partial_1^{\mc H}}{\longrightarrow} C_0(\mc H)$
is the same complex.
\end{lemma}

\begin{proof}
As explained in Section~\ref{section:complex_codes}, from the isomorphism between the space $C_i$ and its dual, the dual complex can be regarded over the spaces $C_i(\mc H)$:
\begin{equation} \label{eqn:dual_comp}
C_2(\mc H) \overset{(\partial_2^{\mc H})^*}{\longleftarrow} C_1(\mc H) \overset{(\partial_1^{\mc H})^*}{\longleftarrow} C_0(\mc H).
\end{equation}
From the one-to-one correspondence between $\mc V$ and $\mc F$, the space $C_0(\mc H)$ and $C_2(\mc H)$ are in bijection, thus we can permute these two spaces in Equation~(\ref{eqn:dual_comp}). Then, using the definition of the transposed application given in Equation~(\ref{eqn:transposed_application}), we recover the original complex.
\end{proof}

Thus the two components of the error can be decoded using the same procedure. In what follows we only consider the $Z$-component of the error. Table~\ref{tab:color_codes} summarizes the properties of errors on color codes. The 2-complex description of the error is deduced immediately from Table~\ref{tab:complex_codes} and Definition~\ref{defi:hypergraph_complex}. 
%The graphical point of view is obtained by considering the error $x \in C_1(\mc H)$ as a subset of $\mc E$. By definition of $\partial_1^{\mc H}$, the syndrome of $x$ is then the set of vertices of $\mc H$ which are touched an odd number of times by the hyperedges of $x$. These vertices are called the terminal vertices of $x$. The space of stabilizers is generated by the vectors of the form $\partial_2^{\mc H}(f)$, for $f \in \mc F$. Since the vector $\partial_2^{\mc H}(f)$ is a hyperface regarded as a set of hyperedges, the stabilizers will be called the sum of hyperfaces.

\begin{table}[h]
\centering
\footnotesize
\caption{Error representation for color codes.}
\label{tab:color_codes}
%\centerline{\footnotesize\smalllineskip
\begin{tabular}{c c c}
\hline
 & 2-complex point of view & graphical point of view\\
\hline
\text{error} & $x \in C_1(\mc H)$ & $x \subset \mc E$ is a subset of the hyperedge set\\
syndrome & $s = \partial_1^{\mc H}(x) \in C_0(\mc H)$ & $s \subset \mc V$ is the set of terminal vertices of $x$\\
stabilizer & $x \in \im \partial_2^{\mc H}$ & $x$ is a boundary\\
\hline
\end{tabular}
\end{table}

The problem of computing a most likely error for a color code can thus be reduced to the determination of a set of hyperedges $x \subset \mc E$, whose terminal vertices are the vertices of $s$, for a given set $s \subset \mc V$. Moreover, we are interested in the determination of $x$ up to the boundaries.

%\begin{fact} \label{fact:color_codes}
%Given a color codes associated to a hypergraph $\mc H=(\mc V, \mc E, \mc F)$, we have:
%
%\begin{itemize}
%\item An error is a vector $x \in C_1(\mc H)$, equivalently, it is a subset of the hyperedge set $x \subset \mc E$.
%\item Its syndrome is the vector $s = \partial_1^*(x)$, equivalently, it is the subset of hypervertices $s \subset \mc V$ which are terminal points of $x$.
%\item A stabilizer is an error $x \in \im \partial_2^*$, equivalently, it is a sum of hyperfaces.
%\end{itemize}
%\end{fact}

The analogous problem for surface codes, stated in Section~\ref{section:surface}, was solved by using a minimum weight perfect matching algorithm. However, the hypergraph structure makes this problem difficult for color codes. For example, the 3-dimensional matching problem is NP-complete \cite{GJ73}.
To decode color codes, our basic idea consists of projecting the error acting on a color code onto three surface codes. Our next goal is to introduce these surface codes and to study the projection onto these codes.

%These surface codes are introduced in section~\ref{subsection:projection} and Proposition~\ref{prop:syndrome_bc} proves that we can compute the syndrome of a projection from the syndrome from a color code error. In Section~\ref{section:lifting}, we explain how to construct an estimation of the color code error, from these three projections. Section~\ref{section:perf} proves that our decoding algorithm allows us to correct the error if its weight is not too large.

\section{The projection as a morphism of 2-complexes}
\label{section:projection}

In this section, we start with some basic properties of the hypergraph complex. Then, we introduce the three surface codes associated with a color code. Our purpose is to transfer the decoding problem for a color code in these surface codes. Thus we have to transfer the whole 2-complex structure of the color code in the surface codes. To this end we will prove that the projection onto the surface codes is a morphism of 2-complexes.

\subsection{Relations between the hypergraph complex and the graph complex}

One of the main difficulty of this section is that we will deal with several different 2-complexes. The first one is the hypergraph chain complex:
$$
C_2(\mc H) \overset{\partial_2^{\mc H}}{\longrightarrow} C_1(\mc H) \overset{\partial_1^{\mc H}}{\longrightarrow} C_0(\mc H)
$$
This 2-complex is denoted by $\mc C(\mc H)$.
A second one is the cellular homology complex associated with the graph $G^*$, denoted by $\mc C(G^*)$:
$$
C_2(G^*) \overset{\partial_2^*}{\longrightarrow} C_1(G^*) \overset{\partial_1^*}{\longrightarrow} C_0(G^*).
$$
This part introduces basic tools to connect these complexes.

The one-to-one correspondence between $\mc E$ and $F^*$ and between $\mc V$ and $V^*$ can be extended to the vectors spaces $C_i$:
\begin{lemma} \label{lemma:bijection}
We have the following equalities:
\begin{itemize}
\item $C_0(\mc H) = C_0(G^*)$.
\item $C_1(\mc H) = C_2(G^*)$,
\end{itemize}
\end{lemma}

The second item of this lemma allows us to apply the transformation $\partial_2^*: C_2(G^*) \rightarrow C_1(G^*)$, of the cellular homology complex of $G^*$, to every vector of $x \in C_1(\mc H)$.
In the following lemma, we compute the image of $x \in C_1(\mc H)$ under $\partial_2^*$, when $x$ corresponds to a hyperedge $e \in \mc E$, and when $x$ corresponds to the boundary of a hyperface $\partial_2^{\mc H} \left( \mc F(v) \right)$.
%Recall that, by definition, $e$ and $\partial_2^{\mc H}$ are subsets of $\mc E$, thus it is also a vector of $C_1(\mc H)$.

\begin{lemma} \label{lemma:images}
With the notations of Equation~(\ref{eqn:hyperedges}) and Equation~(\ref{eqn:hyperfaces}), we have:
\begin{itemize}
\item if $e=\{u_R, u_G, u_B\}$ is a hyperedge of $\mc H$, then
$$
\partial_2^*(e) = \{u_R, u_G\} + \{u_G, u_B\} + \{u_B, u_R\} \in C_1(G^*).
$$
Moreover, the edge $\{u_{\bf c}, u_{\bf c'} \}$ has color $\bf c'' \neq c, c'$.
\item if $f=\mc F(u) = \{e_1, e_2, \dots, e_m\}$ is the hyperface of $\mc H$, then
$$
\partial_2^* \left( \partial_2^{\mc H} (f) \right) = \{v_1, v_2\} + \{v_2, v_3\} + \dots + \{v_{m-1}, v_m\} + \{v_m, v_1\} \in C_1(G^*),
$$
where $e_i = \{u, v_i, v_{i+1}\}$, for $i \in \{1, \dots, m-1\}$, and $e_m = \{u, v_m, v_1\}$.

Moreover, all the edges $\{v_i, v_{i+1}\}$ and $\{v_m, v_1\}$ have the same color. It is the color of $u$.
\end{itemize}
\end{lemma}

\begin{proof}
In the first equation, the hyperedge $e \in \mc E$ is regarded as a face of $G^*$. This face contains the 3 edges $\{u_R, u_G\}, \{u_G, u_B\}$ and $\{u_B, u_R\}$ of $G^*$. Then, the definition of $\partial_2^*$ proves the first item.

To prove the second equality, write $\partial_2^{\mc H}(f) = \sum_i e_i$, and use the linearity of $\partial_2^*$.
To determine the color of the edges of $\partial_2^* \left( \partial_2^{\mc H} (f) \right)$, observe that if $u$ has color $\bf c$, then the vertices $v_i$ share the two other colors alternatively, \emph{i.e.} with $color(v_i) \neq color(v_{i+1})$.
\end{proof}

From this lemma, we observe that the image of a hyperfaces under $\partial_2^* \circ \partial_2^{\mc H}$ is a monochromatic cycle of length $m$. In the next section, we construct a  new tiling, included in $G^*$, using these cycles as faces.
For example, we can see in Figure~\ref{fig:colored_tiling_hyperface}, a hyperface (the hexagon in red) composed of 6 hyperedges (triangles). Its image under $\partial_2^*$ is the red cycle of length 6, which is at the boundary of this hexagon.

\subsection{The three surface codes associated with a color code} \label{section:3_surfaces}

Based on the 3-coloration of the graph $G^*$, we will construct three surface codes. The 3-coloration of the vertices of $G^*$ induces a 3-coloration of the edges. An edge $\{u, v\}$ of $G^*$ inherits of the color which is absent from its endpoints $u$ and $v$. We restrict our attention to the subgraph of $G^*$, induced by the edges colored with $\bf c$, for ${\bf c} \in \{ R, G, B \}$.

\begin{defi} \label{defi:color_subgraph}
The \emph{graph $G^*({\bf c})$} is defined to be the subgraph of $G^*$ induced by the edges colored with ${\bf c}$.
\end{defi}

A graph $G^*$ and its red subgraph $G^*(R)$ are drawn in Figure~\ref{fig:colored_tiling_dual_red}. The following proposition proves that these subgraphs inherit of the tiling structure of the graph $G^*$:

\begin{prop} \label{prop:color_subtiling}
The graph $G^*({\bf c})$, equipped with the set of faces of the form $\partial_2^* \circ \partial_2^{\mc H}(\mc F(v))$, where $v$ is a vertex colored with ${\bf c}$, is a tiling of surface. %Therefore it defines a surface code.
\end{prop}

\begin{proof}
We will prove that $G^*({\bf c})$ is constructed from the tiling of surface $G^*$ by gluing faces of $G^*$. Therefore, it inherits of the structure of tiling of surface of $G^*$.

This graph is connected since $G^*$ is supposed to be connected.
From Lemma~\ref{lemma:images}, the vectors $\partial_2^* \circ \partial_2^{\mc H}(\mc F(v))$ are well cycles.
Now, we have to prove that these cycles of the graph $G^*({\bf c})$ are obtained by gluing the faces of the tiling $G^*$. By definition, a hyperface $\mc F(v)$, is a union of $deg(v)$ disjoint faces of $G^*$ (hyperedges), and a face of $G^*$ appears in exactly one hyperface of the form $\mc F(v)$ such that $v$ is colored with ${\bf c}$.
This concludes the proof.
\end{proof}

In the definition of graphs and tilings of surfaces in Section~\ref{section:tilings}, we assumed that the graph and its dual contain neither loop nor multiple edges. To avoid these configurations in the graph $G^*({\bf c})$ and its dual, it suffices to assume that the length of the shortest cycle which is not a boundary is at least 5.

\begin{figure}[htbp]
\centering
\includegraphics{colored_tiling_dual.pdf}
\includegraphics{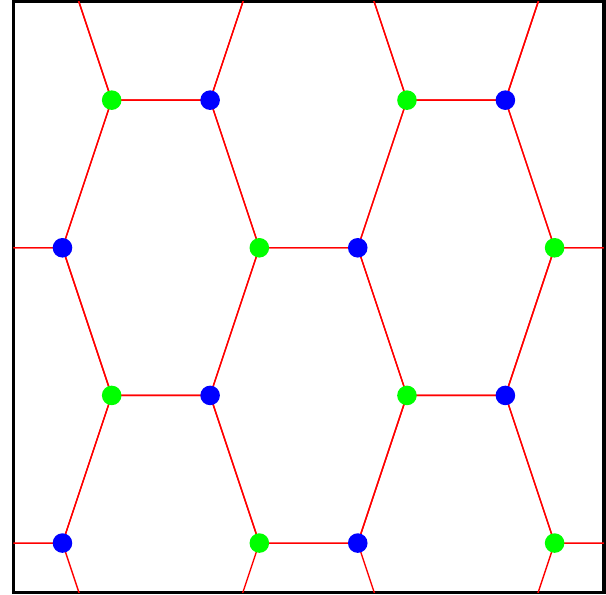}
  \caption{A triangular tiling of the torus $G^*$ with 3-colored vertices and its subtiling $G^*(R)$ induced by the red edges.}
  \label{fig:colored_tiling_dual_red}
\end{figure}

\vspace{.2cm}
From such a tiling $G^*(\bf c)$, we introduce the cellular homology complex $\mc C(G^*(\bf c))$:
$$
C_2(G^*({\bf c})) \overset{\partial_2^{G^*({\bf c})}}{\longrightarrow} C_1(G^*({\bf c})) \overset{\partial_1^{G^*({\bf c})}}{\longrightarrow} C_0(G^*({\bf c})).
$$
This defines the three surface codes associated with a color code. The next part introduces the projection onto these surface codes.

%By construction, this operator allows us to transport from the 2-complex associated with $G^*$ to the 2-complex associated with its subgraph $G^*(\bf c)$:
%\begin{prop} \label{prop:commutativity_restriction}
%The restriction $r_{\bf c}$ is a morphism of 2-complexes from $\mc C(G^*)$ to $\mc C(G^*(\bf c))$, \emph{i.e.} the following diagram is commutative:
%$$
%\begin{CD}
%C_2(G^*) @> \partial_2^* >> C_1(G^*) @> \partial_1^* >> C_1(G^*)\\
%@VV r_{\bf c}^2 V @VV r_{\bf c}^1 V @VV r_{\bf c}^0 V\\
%C_2(G^*({\bf c})) @> \partial_2^{G^*({\bf c})} >> C_1(G^*({\bf c})) @> \partial_1^{G^*({\bf c})} >> C_1(G^*({\bf c}))\\
%\end{CD}
%$$
%\end{prop}
%
%\vspace{.2cm}
%More explicitely, we have $\partial_2^{G^*({\bf c})} \circ r_{\bf c}^2 = r_{\bf c}^{1} \circ \partial_2^*$ and $\partial_1^{G^*({\bf c})} \circ r_{\bf c}^1 = r_{\bf c}^{0} \circ \partial_1^*$.

\subsection{The projection onto the three surface codes} \label{subsection:projection}

Our purpose, in this section, is to introduce a projection from color codes to surface codes and to show that it defines a morphism of 2-complex.
This proves that the projection transfer the whole structure of the color code to the surface code.

%First, we introduce the operator $r_{\bf c}^1$, which is induced by the restriction to the edges of the subgraph $G^*(\bf c)$ of the graph $G^*$.
%The following lemma is stated for convenience. It follows immediately from the partition of the edges of $G^*$ corresponding to the coloration of the edges.
%
%\begin{lemma} \label{lemma:color_decomposition}
%The space $C_1(G^*)$ can be decomposed as:
%$
%C_1(G^*) = \bigoplus_{\bf c} C_1(G^*({\bf c})).
%$
%\end{lemma}
%
%The projections corresponding to this direct sum decomposition are called the restrictions:
%\begin{defi}
%The restriction to the graph $G^*({\bf c})$ is the projection onto the subspace $C_1(G^*_{\bf c})$:
%\begin{align*}
%r_{\bf c}^1: C_1(G^*) & \longrightarrow C_1(G^*(\bf c))\\
%\sum_{e \in E^*} \lambda_e e & \longmapsto \sum_{\substack{e \in E^* \\ \text{color}(e) = {\bf c} }} \lambda_e e.
%\end{align*}
%\end{defi}
%
%These applications $r_{\bf c}^1$ are linear projectors, but, to avoid confusion with the projection introduced below, we prefer the word restriction.

A natural way to define an application from the hypergraph complex to the graph $G^*(\bf c)$ is to send the hypergraph onto the graph $G^*$ and then to restrict ourself to the subgraph $G^*(\bf c)$. For example, to define the first projection which sends $C_1(\mc H)$ onto the space $C_1(G^*(\bf c))$, we send a hyperedge $e \in \mc E$ onto its boundary $\partial_2(e) \in C_1(G^*)$, which is a set of edges of the graph $G^*$. Then, we restrict $\partial_2(e)$ to the edges of $G^*(\bf c)$, that is to say, to the edges colored with $\bf c$. This leads to the following definition.

\begin{defi} \label{defi:projection}
The \emph{projection onto the graph $G^*({\bf c})$} is the triple of operators $\pi_{\bf c} = (\pi_{\bf c}^0, \pi_{\bf c}^1, \pi_{\bf c}^2)$ such that
$$
\pi_{\bf c}^i: C_i(\mc H) \longrightarrow C_i(G^*(\bf c))
$$
and
\begin{eqnarray*}
& \forall v \in \mc V, \quad 
& \pi_{\bf c}^0(v) =
\begin{cases}
v \text{ if } color(v) \neq {\bf c}\\
0 \text{ if } color(v) = {\bf c}\\
\end{cases}
\\
\\
& \forall e = \{v_R, v_G, v_B\} \in \mc E, \quad
& \pi_{\bf c}^1(e) = \{v_{\bf c'}, v_{\bf c''}\},
\text{ with } {\bf c'}, {\bf c''} \neq {\bf c}
\\
\\
& \forall f = \mc F(v) \in \mc F, \quad
& \pi_{\bf c}^2(f) =
\begin{cases}
\partial_2^* \circ \partial_2^{\mc H}(f) \text{ if } color(v) = {\bf c}\\
0 \text{ if } color(v) \neq {\bf c}\\
\end{cases}
\end{eqnarray*}
\end{defi}

The motivation behind this definition of the projection is the following theorem which proves that the projection conserves the structure of the 2-complexes:
\begin{theo} \label{theo:commutativity_projection}
The projection $\pi_{\bf c}$ is a morphism of 2-complexes from $\mc C(\mc H)$ to $\mc C(G^*(\bf c))$, \emph{i.e.} the following diagram is commutative:
$$
\begin{CD}
C_2(\mc H) @> \partial_2^{\mc H} >> C_1(\mc H) @> \partial_1^{\mc H} >> C_1(\mc H)\\
@VV \pi_{\bf c}^2 V @VV \pi_{\bf c}^1 V @VV \pi_{\bf c}^0 V\\
C_2(G^*({\bf c})) @> \partial_2^{G^*({\bf c})} >> C_1(G^*({\bf c})) @> \partial_1^{G^*({\bf c})} >> C_1(G^*({\bf c}))\\
\end{CD}
$$
\end{theo}

\vspace{.2cm}
More precisely, we have
$\partial_2^{G^*(\bf c)} \circ \pi_{\bf c}^2 = \pi_{\bf c}^1 \circ \partial_2^{\mc H}$ and
$\partial_1^{G^*(\bf c)} \circ \pi_{\bf c}^1 = \pi_{\bf c}^0 \circ \partial_1^{\mc H}$.

\begin{proof}
Let us prove that $\partial_1^{G^*(\bf c)} \circ \pi_{\bf c}^1(x) = \pi_{\bf c}^0 \circ \partial_1^{\mc H}(x)$, for all $x$ of $C_1(\mc H)$. Without loss of generality we can assume that ${\bf c} = R$.
By linearity, it is enough to prove the proposition when $x$ corresponds to a hyperedge $e$ of $\mc H$.
This hyperedge is a triple $e=\{v_R, v_G, v_B\}$ of vertices of the hypergraph $\mc H$, the colors of these vertices being indicated by their indices.
By Equation~(\ref{eqn:partial_1_H}), the image of $e$ under $\partial_1^{\mc H}$ is the sum $v_R+v_G+v_B \in C_0(\mc H)$. Then, the application of $\pi_{R}^0$ gives
$\pi_{R}^0 \circ \partial_1^{\mc H} (e) = v_G+v_B$.

Let us compute $\partial_1^{G^*(R)} \circ \pi_{R}^1 (e)$, for $e=\{v_R, v_G, v_B\}$. By Definition~\ref{defi:projection}, the vector $\pi_{R}^1 (e)$ is  $\{v_G,v_B\}$.
Applying $\partial_1^{G^*(R}$ to this vector, we obtain
$\partial_1^{G^*(R)} \circ \pi_{R}^1 (e) = \partial_1^{G^*(R)}(\{v_G,v_B\}) = v_G+v_B$.
This proves that $\partial_1^{G^*(\bf c)} \circ \pi_{\bf c}^1 = \pi_{\bf c}^0 \circ \partial_1^{\mc H}$. The second equality can be proved similarly.
\end{proof}

This Theorem is the key ingredient for our decoding algorithm of color codes. It allows us to transport the 2-complex structure, and therefore the quantum code structure, from the color code to the surface code. This algorithm is described in the next section.

\section{Decoding color codes}
\label{section:decoding}

In this section, we consider a color code, associated with a hypergraph complex $\mc C(\mc H)$.
From Table~\ref{tab:color_codes}, an error is a vector $x \in C_1(\mc H)$. Our goal is to recover $x$, from the knowledge of its syndrome $s=\partial_1^{\mc H}(x)$ and it is enough to identify $x$, up to the space $\im \partial_2^{\mc H}$ of stabilizers.
Our basic idea is to send this error $x$ on surface codes by using the projections introduced in Section~\ref{subsection:projection}. Then, we decode this projected error in the surface codes and we lift the result in the color code.
Our strategy is illustrated in Figure~\ref{fig:couleur_correction} where the different steps of our decoding algorithm are described in a hexagonal color code. We will refer to this example throughout this section.

\subsection{Decoding of the projection of the error}

The projection $\pi_{\bf c}$ transforms the hypergraph complex $\mc C(\mc H)$ in the chain complex $\mc C(G^*({\bf c}))$ introduced in Section~\ref{section:3_surfaces}.
This cellular homology complex defines a surface code.
Recall that, from Table~\ref{tab:surface_codes}, an error on this surface code is a vector $b_{\bf c} \in C_1(G^*({\bf c}))$, its syndrome is the vector $\partial_1^{G^*({\bf c})}(b_{\bf c})$ and the stabilizers are the vectors of $\im \partial_2^{G^*({\bf c})}$.

The projection $\pi_{\bf c}^1(x)$ can be regarded as an error vector on the surface code associated with $G^*(\bf c)$. The following proposition enables us to compute its syndrome in the surface code.

\begin{prop} \label{prop:syndrome_bc}
Let $x \in C_1(\mc H)$ be an error for a color code and let $s \in C_0(\mc H)$ be its syndrome. The projection $\pi_{\bf c}^1(x)$ of the error $x$, is an error for the surface code associated with $G^*({\bf c})$ and its syndrome is the projection $\pi_{\bf c}^0(s)$, of the syndrome of $x$.
\end{prop}

\begin{proof}
From the description of the errors, given in Table~\ref{tab:surface_codes}, the projection $\pi_{\bf c}^1(x)$ is well an error for this surface codes.
Moreover, its syndrome is the vector $\partial_1^{G^*({\bf c})} (\pi_{\bf c}^1(x))$.
By Theorem~\ref{theo:commutativity_projection}, we have
$
\partial_1^{G^*({\bf c})} \circ \pi_{\bf c}^1(x) = \pi_{\bf c}^0 \circ \partial_1^{\mc H}(x).
$
Since the vector $\partial_1^{\mc H}(x)$ is the syndrome of the error $x$, we obtain $\partial_1^{G^*({\bf c})} \circ \pi_{\bf c}^1(x) = \pi_{\bf c}^0(s)$.
\end{proof}

Thanks to this proposition, we are able to compute the syndrome of a projection $\pi_{\bf c}^1(x)$, using only the measured syndrome $s$.
Then, we estimate the projection of the error by applying a surface code decoding algorithm in the tiling $G^*({\bf c})$.
This step is illustrated in Figure~\ref{fig:couleur_correction_restrictionB} and Figure~\ref{fig:couleur_correction_couplageB}.

\subsection{Lifting in the color codes}
\label{section:lifting}

The previous section provides a method to estimate the three projections $\pi_{\bf c}(x)$ of the error $x$ acting on the color code. Our purpose is now to recover the error in the color codes from this information.

First, let us recall a basic property of the chain space $C_1(G^*)$.
\begin{lemma} \label{lemma:C_1_decomposition}
$$
C_1(G^*) = \bigoplus_{\bf c} C_1(G^*(\bf c)).
$$
\end{lemma}

Denote by $\tilde b_{\bf c}$ the estimation of the projection $\pi_{\bf c}(x)$, returned by the surface codes decoding algorithm.
The previous lemma allows us to construct the vector $\tilde b = \sum_{\bf c} \tilde b_{\bf c}$ of $C_1(G^*)$, from the three estimations of the projection.
This vector $\tilde b$ is represented in Figure~\ref{fig:couleur_correction_bord}.
Our goal is to determine $x$ from the knowledge of the vector $\tilde b$.

Our strategy is motivated by the following property of the error.
\begin{prop} \label{prop:b_cycle}
The vector $b = \sum_{\bf c} \pi_{\bf c}(x) \in C_1(G^*)$ is equal to the vector $ \partial_2^*(x)$.
\end{prop}

\begin{proof}
To prove this proposition, let us introduce the projection
$p_{\bf c}: C_1(G^*) \rightarrow C_1(G^*(\bf c))$ corresponding to the direct sum decomposition of Lemma~\ref{lemma:C_1_decomposition}.
By definition~\ref{defi:projection}, we have $\partial_{\bf c}^1 = p_{\bf c} \circ \partial_2^*$. Therefore, the vector $b$ can be written $b = \sum_{\bf c} p_{\bf c} \circ \partial_2^*(x)$. Then remark that the sum $\sum_{\bf c} p_{\bf c}$ is the identity operator. This proves the proposition.
\end{proof}

In other words, Proposition~\ref{prop:b_cycle} proves that $b$ is a boundary in the tiling $G^*$. Our strategy is to fill the estimation $\tilde b$ of the vector $b$, when it is possible. This corresponds to Figure~\ref{fig:couleur_correction_remplissage}.
When the error has sufficiently low weight, the estimation $\tilde b$ is a boundary and its filling allows us to recover the original error $x$ for the color code, up to a stabilizer. This fact is proved in the next section.

\begin{figure}[htbp]
\centering

\subfloat[An error $x$ whose support is the set of grey hyperedges.
Its syndrome $s$ is composed of the vertices supporting a symbol~1.]
{\includegraphics[scale=.7]{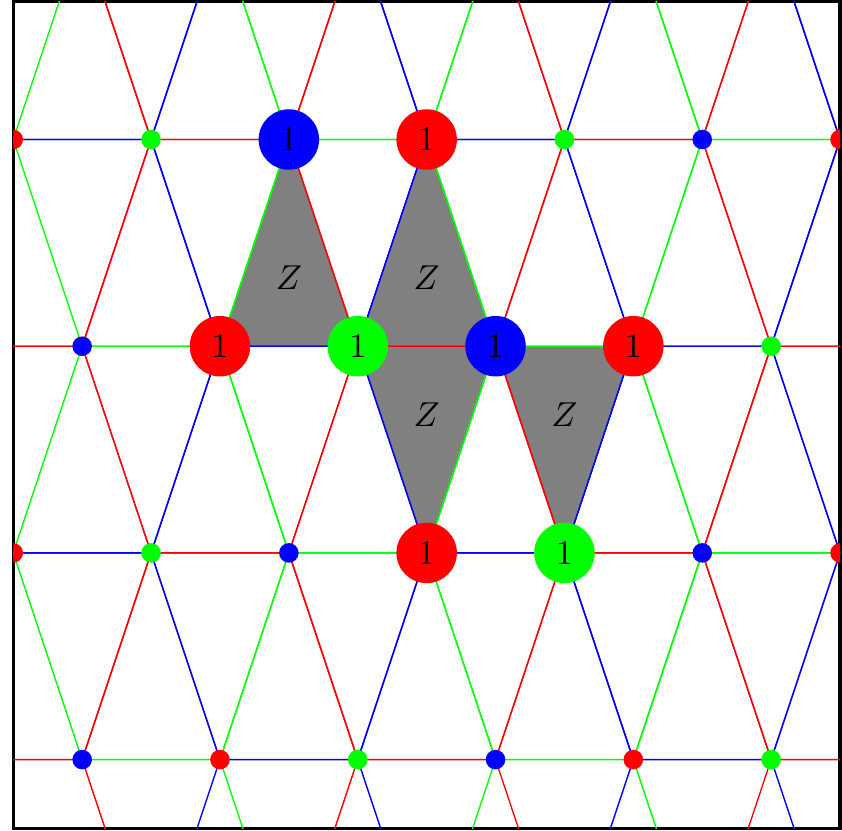}}
\hspace{.5cm}
\subfloat[We restrict the syndrome to the blue graph $G^*(B)$]
{\label{fig:couleur_correction_restrictionB}\includegraphics[scale=.7]{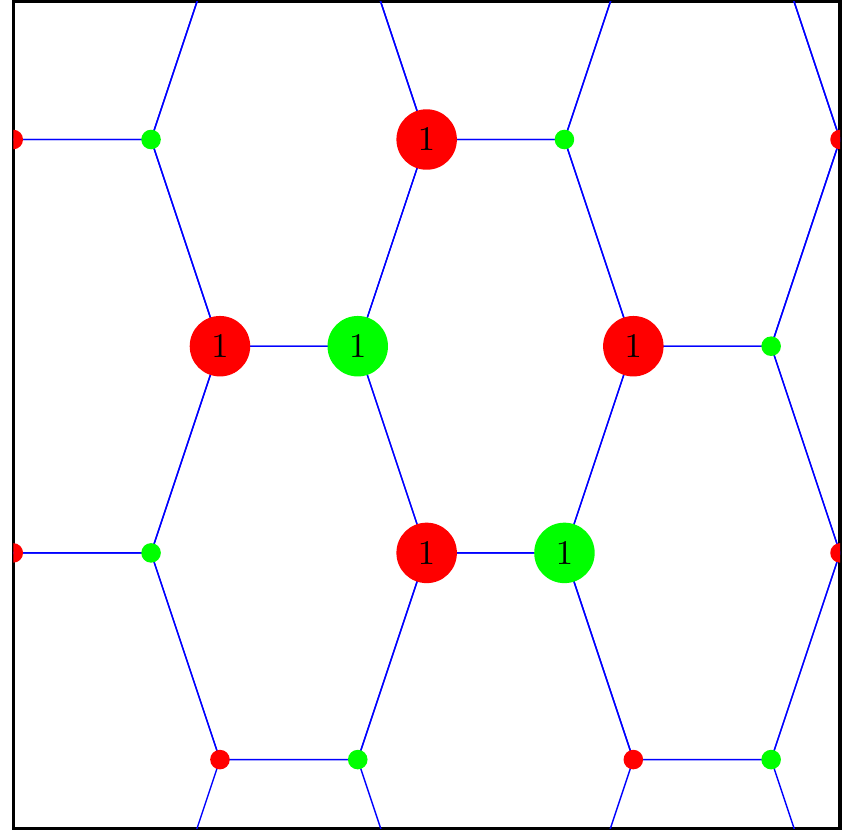}}

\subfloat[We decode the projection on the surface code $G^*(B)$. This defines a set $\tilde b_B$ of dotted blue edges.]
{\label{fig:couleur_correction_couplageB}\includegraphics[scale=.7]{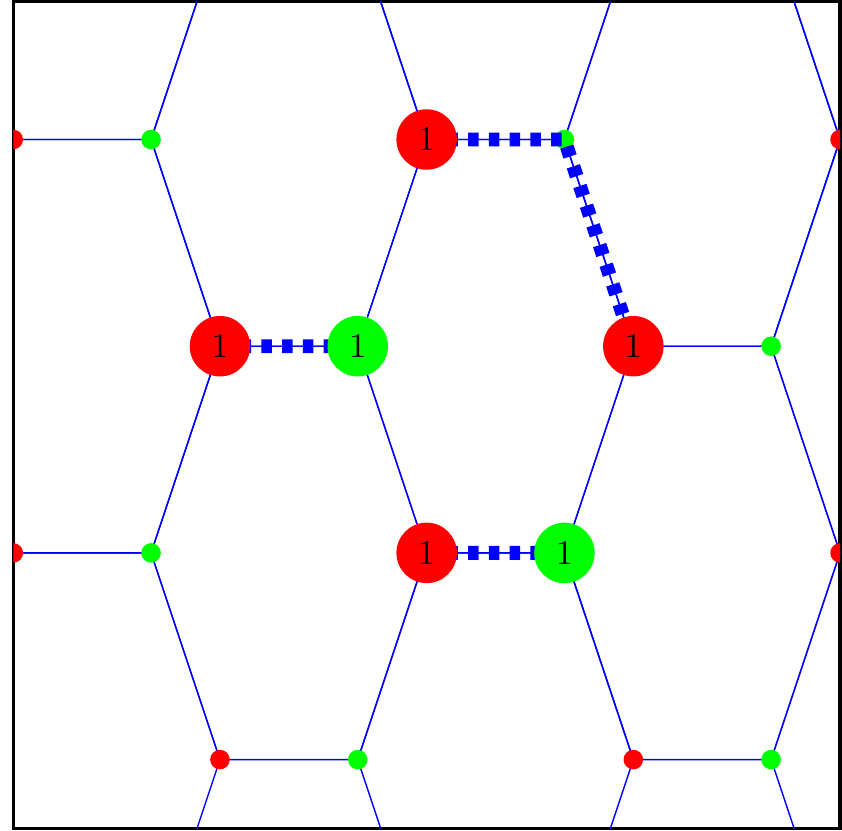}}
\hspace{.5cm}
\subfloat[We consider the union $\tilde b$ of the three sets $\tilde b_{\bf c}$ of edges computed in the three surface codes associated with graph $G^*({\bf c})$.]
{\label{fig:couleur_correction_bord}\includegraphics[scale=.7]{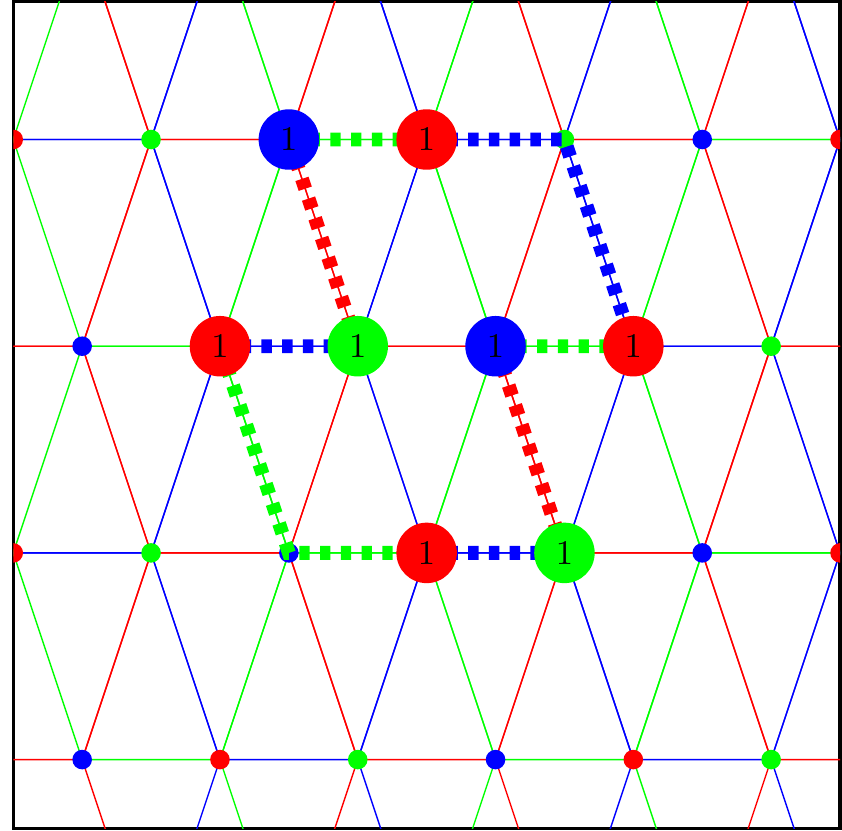}}

\subfloat[The estimation $\tilde x$ of the error is the set of grey edges obtained by filling the cycle~$\tilde b$.]
{\label{fig:couleur_correction_remplissage}\includegraphics[scale=.7]{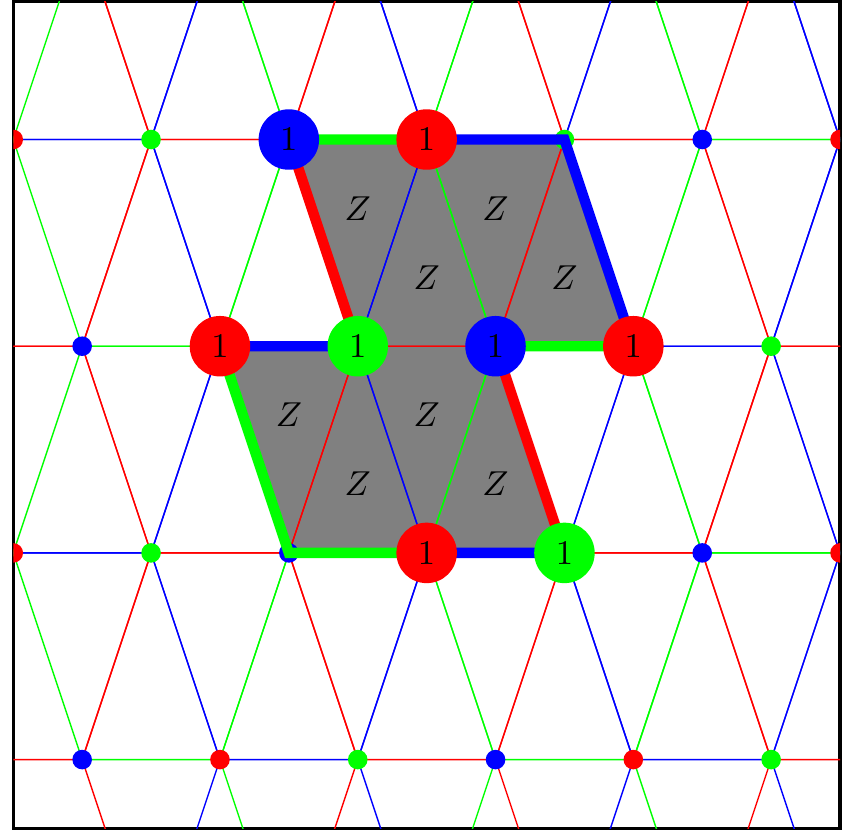}}
\hspace{.5cm}
\subfloat[The sum $x + \tilde x$ corresponds to the grey hyperedges. It is a sum of two hyperfaces therefore the vector $\tilde x$ is a good estimation of $x$.]
{\includegraphics[scale=.7]{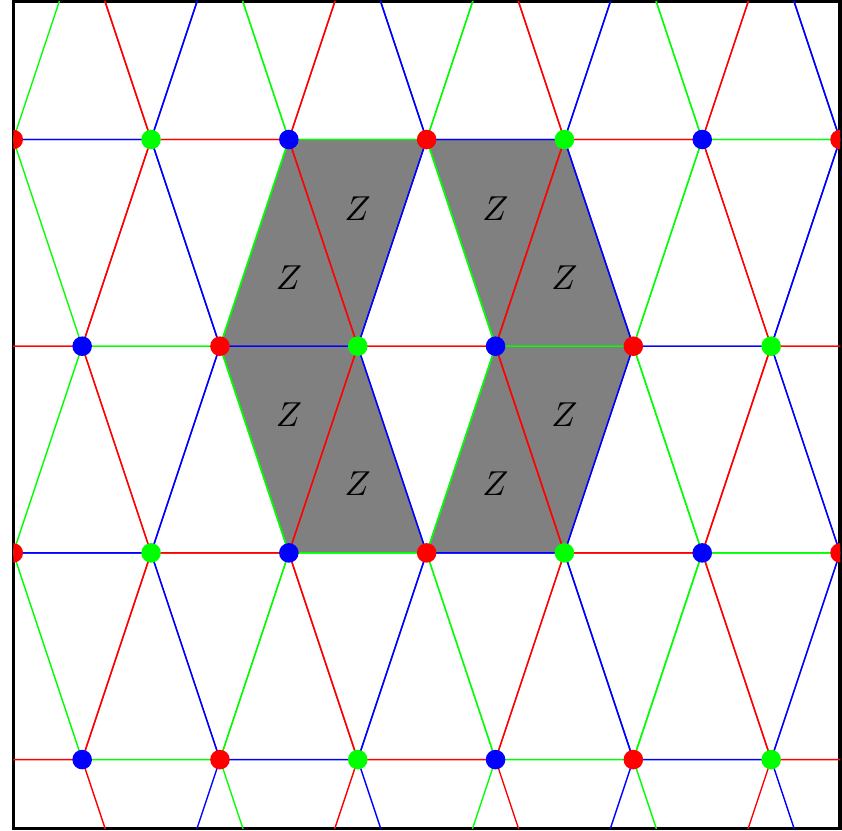}}

  \caption{An example of decoding of a color code}
\label{fig:couleur_correction}
\end{figure}

\begin{algorithm}
\caption{Decoding algorithm for color codes}
\label{algo:couleur}

\begin{algorithmic}[1]
\REQUIRE $s \in \im \dha$ the syndrome of the error $x$.
\ENSURE $\tilde x$ such that $\dha(\tilde x)=s$ and $\partial_2^*(\tilde x)$ of minimum weight or $NO LIFTING$.

\STATE Compute the projection $\pi_{\bf c}^0(s)$ onto $G^*({\bf c})$ for ${\bf c} \in \{R, G, B\}$.

\STATE Decode the surface code associated with $G^*({\bf c})$, using the syndrome $\pi_{\bf c}^0(s)$, for ${\bf c} \in \{R, G, B\}$. This returns the three vectors $\tilde b_{\bf c}$. Compute $\tilde b = \sum_{{\bf c} \in \{R, G, B\}} \tilde b_{\bf c}$.

\STATE Determine $\tilde x \in C_2(G^*)$ such that $\partial_2(\tilde x) = \tilde b$, by filing the boundary $\tilde b$ in the tiling $G^*$.
If $\tilde b$ cannot be filled return $NO LIFTING$.

\end{algorithmic}
\end{algorithm}

The complexity of Algorithm~\ref{algo:couleur}, which allows us to decode color codes is polynomial.
The most expensive step of this algorithm is the surface decoding algorithm, when we use the perfect matching algorithm.
Algorithm~\ref{algo:couleur} could be parallelized by combining the improved perfect matching decoding of Fowler \emph{et al.} \cite{FWH12}, with an approximation of the lifting of $\tilde b$.

\subsection{Performance of our decoding algorithm for color codes}
\label{section:perf}

The following theorem proves that our color codes decoding algorithm corrects an error for the color code when the three surface codes decoding algorithms correct the projected errors.

Let $x\in C_1(\mc H)$ be an error for a color code and let $\pi_{\bf c}(x)$ be its surface code projection.
Recall that a surface code decoding algorithm allows us to correct an error $\pi_{\bf c}^1(x)$, iff it returns a vector $\tilde b_{\bf c}$, which is a vector of the coset of $\pi_{\bf c}^1(x)$, modulo the subspace $\im \partial_2^{G^*({\bf c})}$.
Similarly, we say that a color codes decoding algorithm corrects an error $x \in C_1(\mc H)$ iff it finds a vector $\tilde x \in C_1(\mc H)$, which is equivalent to $x$, modulo the subspace $\im \dha$.

\begin{theo} \label{theo:perf_algo2}
If the surface code decoding algorithm corrects the projection $\pi_{\bf c}^1(x)$ of the error $x$ for the three colors ${\bf c} = R, G$ and $B$, then Algorithm~\ref{algo:couleur} corrects the error $x$.
\end{theo}

\begin{proof}
Assume that $\tilde b_{\bf c}$ is a vector of the coset of $\pi_{\bf c}(x)$ modulo the subspace $\im \partial_2^{G^*({\bf c})}$ for the three colors ${\bf c}$. We must prove that the vector $\tilde b = \sum_{\bf c} \tilde b_{\bf c}$ of $C_1(G^*)$ belongs to the space $\im \partial_2^*$ and that all its premiages $\tilde x$ are in the coset of $x$ modulo $\im \partial_2^{\mc H}$.
By linearity, it is enough to show that, if $\tilde b_{\bf c} \in \im \partial_2^{G^*({\bf c})}$ for all $\bf c$, then $\tilde b \in \im \partial_2^*$ and its preimages under $\partial_2^*$ belong to the space $\im \partial_2^{\mc H}$.

Denote by $\pi^i$ the linear application
$$
\pi^i : C_i(\mc H) \longrightarrow \prod_{\bf c} C_i(G^*({\bf c})),
$$
which maps $x$ to the triple $(\pi_R^i(x), \pi_G^i(x), \pi_B^i(x))$.
From Theorem~\ref{theo:commutativity_projection}, the following diagram is commutative:

\vspace{.2cm}
$$
\begin{CD}
C_2(\mc H) @> \partial_2^{\mc H} >> C_1(\mc H)\\
@VV \pi^2 V @VV \pi^1 V\\
\prod_{\bf c} C_2(G^*({\bf c})) @> \prod_{\bf c} \partial_2^{G^*({\bf c})} >> \prod_{\bf c} C_1(G^*({\bf c}))\\
\end{CD}
$$
\vspace{.2cm}

The set $\prod_{\bf c} C_1(G^*({\bf c}))$ is in one-to-one correspondence with the direct sum $\bigoplus_{\bf c} C_1(G^*(\bf c))$. Then, By Lemma~\ref{lemma:C_1_decomposition}, it is also in one-to-one correspondence with the space $C_1(G^*)$.
Denote by $\phi: C_1(G^*) \rightarrow C_1(G^*({\bf c}))$ this bijection. Then, the application $\phi \circ \partial_2^*$ coincides with the application $\partial_1$.
Thus, it is sufficient to show that, if we have $\tilde b_{\bf c} \in \im \partial_2^{G^*(\bf c)}$ for all $\bf c$, then the triple $(\tilde b_R, \tilde b_G, \tilde b_B)$ admits a preimage under $\pi^1$, and that the preimages of this triple are included in the space $\im \partial_2^{\mc H}$.

Therefore, we assume that $\tilde b_{\bf c} \in \im \partial_2^{G^*(\bf c)}$, for all $\bf c$. For each color $\bf c$, there exists a vector $\tilde \beta_{\bf c} \in C_2(G^*(\bf c))$ such that $\partial_2^{G^*({\bf c})}(\tilde \beta_{\bf c}) = \tilde b_{\bf c}$.
The triple $(\tilde \beta_R, \tilde \beta_G, \tilde \beta_B)$ is then a preimage of $(\tilde b_R, \tilde b_G, \tilde b_B)$ under the application $\prod_{\bf c} \partial_2^{G^*({\bf c})}$.
To lift the vector $(\tilde \beta_R, \tilde \beta_G, \tilde \beta_B)$ in the space $C_2(\mc H)$, it suffices to remark that the application $\pi^2$ is an isomorphism. Therefore, there exists a vector $\alpha \in C_2(\mc H)$ such that
$$
\left( \prod_{\bf c} \partial_2^{G^*({\bf c})} \right) \circ \pi^2 (\alpha) = (\tilde b_R, \tilde b_G, \tilde b_B).
$$
By commutativity of the diagram drawn below, we have
$
\pi^1 \circ \partial_2^{\mc H} (\alpha) =  (\tilde b_R, \tilde b_G, \tilde b_B).
$
This proves that $\partial_2^{\mc H} (\alpha)$ is a preimage of $(\tilde b_R, \tilde b_G, \tilde b_B)$, under the application $\pi^1$.
Therefore, the triple $(\tilde b_R, \tilde b_G, \tilde b_B)$ is in the set $\im \partial_2^*$ and its preimage is clearly in the image of $\partial_2^{\mc H}$.

To conclude, let us prove that the other preimages of the vector $\tilde b$ are also included in the space $\im\partial_2^{\mc H}$. First, consider the kernel of the linear application $\pi^1$. If a set $x$ of hyperedges satisfies $\pi^1(x) = 0$, then it corresponds to a set of faces of the graph $G^*$, which has an empty boundary. By connexity of the tiling $G$, it is either zero or the set of all the faces of $G^*$. Thus, we have
$
\ker \pi^1 =\{ 0, \sum_{e \in \mc E} e \}.
$
Consequently, the vector $\tilde b$ has two preimages and that their difference is the vector $\sum_{e \in \mc E} e$. Since this sum is in the image of the application $\partial_2^{\mc H}$, the second preimages of $\tilde b$ is also included in the space $\im \partial_2^{\mc H}$.
%The vector $\tilde b$ has another preimage. The fact that this second vector is also included in the space $\im\partial_2^{\mc H}$ can be derived from Proposition~\ref{prop:components_graph}.
\end{proof}

For numerical simulations, we consider the hexagonal color codes studied by Sarvepalli and Raussendorf in 2012 \cite{SR12}. Let us recall the definition of these color codes.
Denote by $H_r$ the Cayley graph of the group $\Z/(3r)\Z \times \Z/(3r)\Z$ and the generating set:
$$
\{ \pm (1,0), \pm (0, 1), \pm (1,-1) \}.
$$
By definition, it is the graph whose vertices are the elements of the group and two vertices are joined by an edge iff they differ in an element of the generating set.
This graph can be naturally embedded in the torus. It is endowed with triangular faces and its vertices can be 3-colored.
The dual tiling of $H_r$ is a trivalent hexagonal tiling, equipped with 3-colored faces. It defines a color code of parameters $[[18.r^2, 4,  4r]]$.

The performance of Algorithm~\ref{algo:couleur} for the color codes based on $H_r$, with $r=2^m$, is represented in Figure~\ref{fig:courbe_color_hexagonal}. We observe an error threshold close to $8.7\%$. The decoding algorithm proposed by Sarvepalli and Raussendorf gave a threshold of $7.8\%$.

\begin{figure}[htbp]
\centering
\includegraphics[scale=1]{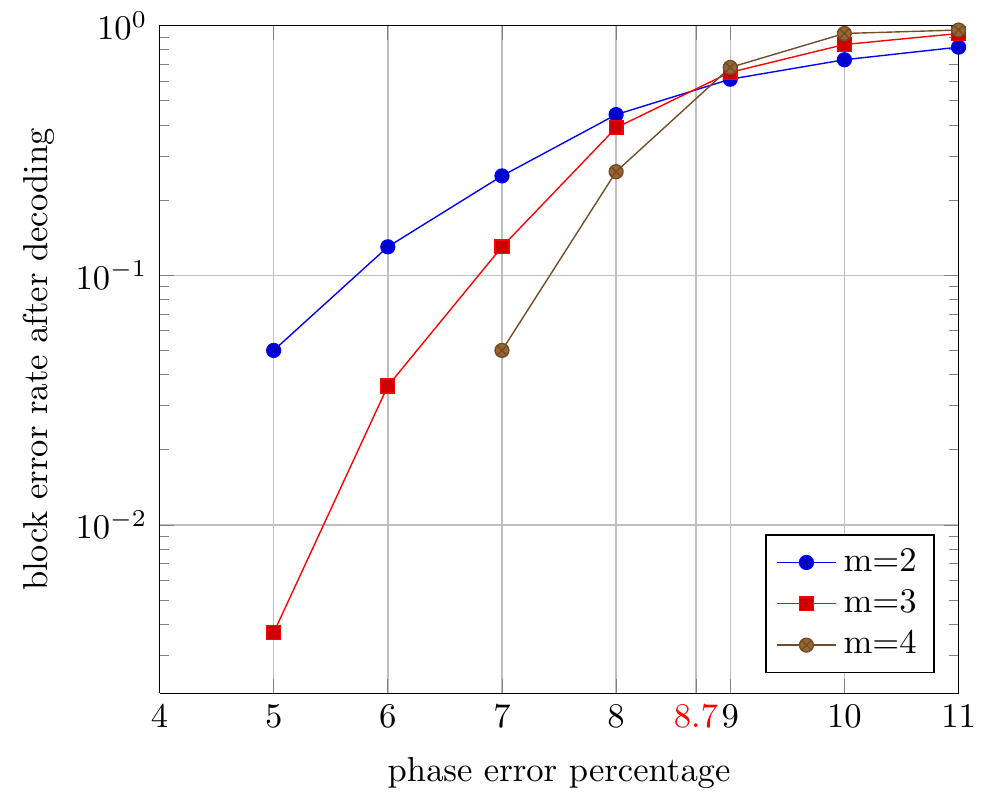}
  \caption{Bloc error rate after decoding as a function of the phase (or bit flip) error rate, for hexagonal color codes of parameters $[[18.4^m, 4, 4.2^m]]$.}
  \label{fig:courbe_color_hexagonal}
\end{figure}

\subsection{Comparison of the thresholds of color codes and surface codes}
\label{subsection:threshold}

The projection morphism, studied in Section~\ref{section:projection}, enables us to transfer results from surface codes to color codes. In this section, we compare the error threshold of surface codes and color codes.

%The {\em depolarizing error threshold} of a family of stabilizer codes $(Q_t)_t$ is defined to be the highest depolarizing probability $q$, that can be tolerate with vanishing error probability after decoding when $t\rightarrow\infty$. This threshold depends on the decoding algorithm.
In this paper, we consider independently the bit flip error $X$ and the phase error $Z$. Thus, we study the {\em phase error threshold} and the {\em bit flip error threshold}. The phase error threshold of a family of stabilizer codes $(Q_t)_t$ is defined to be the highest phase error probability, that can be tolerate with vanishing phase error probability after decoding when $t\rightarrow\infty$.

By symmetry, the phase error threshold of a family of color codes coincides with its bit flip error threshold. For a family of surface codes, these thresholds can be different. It is the case for example with the surface codes based on the triangular lattice and the hexagonal lattice. The phase error threshold is observed in the original graph and the bit flip error threshold is observed in the dual graph.

Let us consider a family of color codes, denoted by $(\mc C_t)_t$, and the three corresponding surface codes, denoted by $(\mc C_t({\bf c}))_t$, for ${\bf c} \in \{R, G, B\}$.
Assume that we use a surface decoding algorithm and the color decoding algorithm deduced from it using Algorithm~\ref{algo:couleur}.
Given these decoding algorithms, denote by $p_c$ the phase error threshold of the family of color codes $(\mc C_t)_t$ and denote by $p_c(\bf c)$ the phase error threshold of the family of surface codes $(\mc C_t(\bf c))_t$.

\begin{theo} \label{theo:comparaison_seuil}
The phase error threshold $p_c$ of a family of color codes is bounded below by a function of the phase error threshold $p_c(\bf c)$ of the three corresponding families of surface codes:
$$
p_c \geq \min_{{\bf c} \in \{R, V, B\}} \left\{ \frac 1 2\left( 1 - \sqrt{1-2p_c({\bf c})} \right) \right\}.
$$
\end{theo}

\begin{proof}
Consider an error $x \in C_1(\mc H)$ for the color code and its projection $\pi_{\bf c}(x)$, which is an error in the surface code corresponding to the color $\bf c$.
When $x$ is a random error for a binary symmetric channel of probability $p$, then $\pi_{\bf c}(x)$ is a random error for a binary symmetric channel of probability $2p(1-p)$. To prove this, remark that every qubit of the surface code corresponds to two qubits of the color code. Thus, the probability to have an error on a qubit of the surface code is the probability to have an error on exactly one of the corresponding two qubits in the color code.

If the probability $2p(1-p)$ is below the threshold $p_c(\bf c)$ of the three surface codes, then we are able to correct the error $\pi_{\bf c}(x)$, in the three surface codes, with vanishing error probability, using the surface code decoding algorithm. From Theorem~\ref{theo:perf_algo2}, we can also correct the error $x$ acting on the color code with vanishing error probability, in this case.
This proves that $p$ is below the phase error threshold of the family of color codes $(C_t)_t$. To conclude the proof, it remains only to solve the equation $2p(1-p)=p_c(\bf c)$.
\end{proof}

In the special case of the hexagonal color codes, the three corresponding surface codes are based on an hexagonal tiling of the torus. The performance of the perfect matching algorithm over these surface codes was simulated in Figure~\ref{fig:perf_surface}, where we found a phase error threshold close to $15.9\%$.
The phase error threshold of hexagonal color codes obtained in Figure~\ref{fig:courbe_color_hexagonal} is approximately $8.7\%$.
These numerical results are in excellent agreement with Theorem~\ref{theo:comparaison_seuil}, indeed when $p_c({\bf c})\approx 0.159$, we have $\frac 1 2\left( 1 - \sqrt{1-2p_c({\bf c})} \right) \approx 0.0871$.

\section*{Conclusion}

\begin{itemize}

\item We proposed a new decoding algorithm for color codes. This algorithm is based on the projection of the error onto three surface codes. This strategy allows us to transform every decoding algorithm for surface codes, such as perfect matching decoding \cite{DKLP02}, or renormalization group decoding \cite{DP10} into a decoding algorithm for color codes.
Using the perfect matching decoding algorithm for surface codes, we found a threshold of $8.7\%$ for a family of hexagonal color codes. This value is higher than the threshold observed by Sarvepalli and Raussendorf for these codes \cite{SR12}.

\item
Our idea shares some common features with the method recently proposed by Bombin, Duclos-Cianci and Poulin, which is based on the decomposition of a color codes into multiple copies of toric codes \cite{BDP12}. The point of view developed here is more general because it can be applied to every kind of color code, not only those based on square tilings. Recall that the use of square tilings is a strong restriction, because it makes impossible the construction of positive rate topological codes with growing minimum distance \cite{BPT10, De13}.
For algorithmic considerations, we deal with only three surface codes whereas the algorithm of Bombin \emph{et al.} can involve a large number of copies of the toric code.

\item
From a theoretical point of view, this study of the decoding problem of color codes could be used to transfer results from surface codes to color codes. For example, using the fact that our algorithm conserves the error model, we connected the threshold of color codes and surface codes in Theorem~\ref{theo:comparaison_seuil}.
%The study of the threshold of quantum LDPC codes, which are highly degenerate codes, is particularly important because it could provide some information about the capacity of the depolarizing channel.

\item
One future direction is optimizing our decoding algorithm, which is not always able to estimate the most likely error. Indeed, when the boundary estimation $\tilde b$ is not a vector of $\im \partial_2^*$, Algorithm~\ref{algo:couleur} returns $NO LIFTING$.
One could examine these problematic configurations. Further improvement may also be possible by using the correlations between the three projections of the error on the surface codes.
Finally, we do not consider the correlation between phase errors and bit flip errors for the depolarizing channel in this paper. The study of these correlations is necessary to approach the optimal threshold of the depolarizing channel \cite{BAOKM12}.
\end{itemize}

\section*{Acknowledgments}

The author wishes to acknowledge Alain Couvreur, Gilles Zémor and Benjamin Smith for careful reading of this paper. This work was supported by the French ANR Defis program under contract ANR-08-EMER-003 (COCQ project), and the CNRS PEPS project TOCQ. Part of this work was done while the author was supported by the Délégation Générale pour l'Armement (DGA) and the Centre National de la Recherche Scientifique (CNRS) at Bordeaux University. He is now supported by the LIX-Qualcomm Postdoctoral fellowship at the \'Ecole Polytechnique.

\appendix
\makeatletter
\def\@seccntformat#1{Appendix~\csname the#1\endcsname:\quad}
\makeatother

\section{Lifting of a boundary in a tiling of surface}

The goal of this section is to describe an algorithm to determine a vector $\tilde x \in C_2(G^*)$, such that $\partial_2^*(\tilde x) = \tilde b$, given a vector $\tilde b \in C_1(G^*)$. In the original graph $G$, the vector $\tilde x$ correspond to a set $X$ of vertices and the vector $\tilde b$ corresponds to a set $B$ of edges. For $X \subset V$, we denote by $\partial(X)$ the set of edges leaving $X$, \emph{i.e.} the set of edges having exactly one end-point in $X$. This set $\partial(X)$ is called the boundary of $X$.
This leads to the following formulation of the problem: given a set of egdes $B$ of the graph $G$, determine a set $X$ of vertices of the graph $G$, whose boundary $\partial(X)$ is equal to $B$.

First, we must be able to distinguish the connected components of the graph $G$, limited by the edges of $B$. The subgraph of $G$, induced by the edges of $E \backslash B$, is denoted by $G_{B^c}$. We are interested in the connected components of the graph $G_{B^c}$.

\begin{defi}
Let $C_1, C_2, \dots C_\kappa$ be the connected components of the graph $G_{B^c}$. The \emph{components graph} of $G_{B^c}$ is the graph of vertex set $V=\{1, 2, \dots, \kappa\}$ whose edges are the pairs of vertices $\{i, j\}$, such that the subgraphs $C_i$ and $C_j$ of $G$ are linked by an edge of $B$.
\end{defi}

%By construction of $\mc H$, there is a one-to-one correspondence between the vertices of $G$ and the hyperedges of $\mc H$. Therefore, a subset $U$ of vertices of $G$ corresponds to a vector $x(U) \in C_1(\mc H)$, which is defined as $x(U)=\sum_{u \in U} \mc E(u)$, where $\mc E(u)$ is the hyperedge of $\mc H$ corresponding to the vertex $u\in V(G)$.

To find the desired set of vertices $X$, we choose some of the connected components of the graph $G_{B^c}$. We must be careful to not take two neighbors components. Otherwise, we do not get the edges of $B$ which separate these components in the boundary $\partial(X)$. The next proposition characterizes the sets $B$, that can be lifted into a satisfying set $X$, and gives a method to construct the lifting $X$.

\begin{prop} \label{prop:components_graph}
Let $\mc C_{B^c}$ be the components graph of $G_{B^c}$.
Then, there exists a set $X \subset V$ such that $\partial(X) = B$ if and only if the graph $G_{B^c}$ is a bipartite graph.
Moreover, if $G_{B^c}$ is a bipartite graph, denote by $V_1 \cup V_2$ the corresponding partition of the vertices. Then the sets $X \subset V$ such that $\partial(X) = B$ are the two sets
\begin{equation} \label{eqn:lifting}
X_i = \{ v \in V \ | \ v \in C_j \text{ and } j \in V_i \},
\end{equation}
for $i= 1$ and $i=2$.
\end{prop}

\begin{proof}
Suppose that the set $B$ admits a preimage $X$ such that $\partial(X) = B$.
Let us prove that we can partition the vertices of the components graph $G_{B^c}$.
Let $C_i$ and $C_j$ be two connected components of the graph $G_{B^c}$, joined by an edge $e=\{u, v\}$ of $B$, \emph{i.e.} $u \in C_i$ and $v \in C_j$.
The set $X$ contains exaclty one of the end-points of $e$. Otherwise, the edge $e \in B$ does not belong to the set $\partial(X)$. Assume that $u \in X$ and $v \notin X$.
Since the set $X$ contains a vertex $u$ of the component $C_i$, it also contains all the neighbors of $x$, which are included in $C_i$. Step by step, this proves that all the vertices of this connected component are in the set $X$.
Moreover, there is no vertex of the connected component $C_j$ included in the set $X$. Indeed, if $X$ contains a  vertex of $C_j$, then it contains the whole component $C_j$, and we have $v \in X$. This contradicts the fact that $X$ contains exaclty one of the end-points of $e$.

Finally, we proved that the set $X$ is a union of connected components $C_i$, such that if we consider two neighbors connected components, it contains exactly one of these two components. This proves that the vector $X$ defines a partition of the vertices of the components graph corresponding to a bipartite structure.
The converse statement can be proved similarly
\end{proof}

\begin{algorithm}
\caption{Lifting of a boundary}
\label{algo:lifting}

\begin{algorithmic}[1]
\REQUIRE a graph $G = (V, E)$, a subset $B \subset E$.
\ENSURE a subset $X \subset V$ such that $\partial(X) = B$, or $NO LIFTING$ when such a set $X$ does not exist.

\STATE Compute the connected components of the graph $G_{B^c}$.

\STATE Define the components graph $\mc C_{B^C}$ of the graph $G_{B^c}$.

\STATE Compute a bipartite decomposition $V_1 \cup V_2$ of $\mc C_{B^c}$. If this graph is not bipartite return $NO LIFTING$.

\STATE Return $X_1 = \{ v \in V \ | \ v \in C_j \text{ and } j \in V_1 \}$.
\end{algorithmic}
\end{algorithm}

From this proposition, we able to determine if a set $B \subset E$ has a lifting $X \subset V$ such that $\partial(X) = B$, and we can compute efficiently this lifting $X$, when such a lifting exists. This result leads to Algorithm~\ref{algo:lifting}.
The complexity of this algorithm is quadratic due the determination of the connected components of a graph \cite{HT73}.
%can be done in $O(|E||V|)$ \cite{HT73}.

\bibliographystyle{acm}
%\bibliography{biblio}

\newcommand{\SortNoop}[1]{}

\end{document}